\documentclass[11pt]{article}
\usepackage{graphicx}
\usepackage{latexsym}
\usepackage{amssymb}
\usepackage{amsmath}
\usepackage{amsthm}
\usepackage{forloop}
\usepackage[paper=letterpaper,margin=1in]{geometry}
\usepackage[ruled,vlined,linesnumbered]{algorithm2e}
\usepackage{paralist}

\usepackage{subfigure}
\usepackage{bm}
\usepackage[capitalize,noabbrev]{cleveref}
\usepackage{xspace}
\DeclareMathOperator{\supp}{supp}

\newcommand{\Alg}{{\texttt{Residual Random Greedy}}\xspace}
\newcommand{\AlgGreedy}{{\texttt{Standard Greedy}}\xspace}
\newcommand{\AlgRand}{{\texttt{Random}}\xspace}

\newtheorem{theorem}{Theorem}[section]
\newtheorem{lemma}[theorem]{Lemma}

\newtheorem{corollary}[theorem]{Corollary}
\newtheorem{definition}{Definition}[section]

\newtheorem{observation}[theorem]{Observation}

\makeatletter
\newtheorem*{rep@theorem}{\rep@title}
\newcommand{\newreptheorem}[2]{%
\newenvironment{rep#1}[1]{%
 \def\rep@title{#2 \ref{##1}}%
 \begin{rep@theorem}}%
 {\end{rep@theorem}}}
\makeatother

\newreptheorem{theorem}{Theorem}
\newreptheorem{reduction}{Reduction}

\newcommand{\defcal}[1]{\expandafter\newcommand\csname c#1\endcsname{{\mathcal{#1}}}}
\newcommand{\defbb}[1]{\expandafter\newcommand\csname b#1\endcsname{{\mathbb{#1}}}}
\newcounter{calBbCounter}
\forLoop{1}{26}{calBbCounter}{
    \edef\letter{\Alph{calBbCounter}}
    \expandafter\defcal\letter
		\expandafter\defbb\letter
}

\newcommand{\eps}{\varepsilon}
\newcommand{\ie}{{\it i.e.}}

\newcommand{\nnR}{{\bR_{\geq 0}}}

\newcommand{\RRG}{\Alg}

\title{Weakly Submodular Maximization Beyond Cardinality Constraints: Does Randomization Help Greedy?}
\author{
Lin Chen\thanks{Yale Institute for Network Science, Yale University.  E-mail: \texttt{lin.chen@yale.edu}. Authors are listed in alphabetical order.}
\and
Moran Feldman\thanks{Depart. of Mathematics and Computer Science, The Open University of Israel.
E-mail: \texttt{moranfe@openu.ac.il}.}
\and
Amin Karbasi\thanks{Yale Institute for Network Science, Yale University. E-mail: \texttt{amin.karbasi@yale.edu}.}
}
\begin{document}
\maketitle
\begin{abstract}
	Submodular functions are a broad class of set functions, which naturally arise in diverse areas such as economics, operations research and game theory. Many algorithms have been suggested for the maximization of these functions, achieving both strong theoretical guarantees and good practical performance. Unfortunately, once the function deviates from submodularity (even slightly), the known algorithms for submodular maximization may perform arbitrarily poorly. Amending this issue, by obtaining approximation results for classes of set functions generalizing submoduolar functions, has been the focus of several recent works.
	
	One such class, known as the class of weakly submodular functions, has received a lot of attention from the machine learning community due to its strong connections to restricted strong convexity and sparse reconstruction. A key result proved by Das and Kempe~(2011) showed that the approximation ratio of the standard greedy algorithm for the problem of maximizing a weakly submodular function subject to a cardinality constraint degrades smoothly with the distance of the function from submodularity. However, no results have been obtained so far for the maximization of weakly submodular functions subject to constraints beyond cardinality. In particular, it is not known whether the greedy algorithm achieves any non-trivial approximation ratio for such constraints.
	
	In this paper, we prove that a randomized version of the greedy algorithm (previously used by Buchbinder et al.~(2014) for a different problem) achieves an approximation ratio of $(1 + 1/\gamma)^{-2}$ for the maximization of a weakly submodular function subject to a general matroid constraint, where $\gamma$ is a parameter measuring the distance of the function from submodularity. Moreover, we also experimentally compare the performance of this version of the greedy algorithm on real world problems such as gene splice site detection and video summarization against natural benchmarks, and show that the algorithm we study performs well also in practice. To the best of our knowledge, this is the first algorithm with a non-trivial approximation guarantee for maximizing a weakly submodular function subject to a constraint other than the simple cardinality constraint. In particular, it is the first algorithm with such a guarantee for the important and broad class of matroid constraints.
	
\medskip
\noindent \textbf{keywords:} weakly submodular functions, optimization, matroid constraint
\end{abstract}

\thispagestyle{empty}
\newpage
\pagenumbering{arabic}

\section{Introduction}

Motivated by the frequent appearances of submodular functions in both theoretical and practical settings, the last decade has seen a proliferation of works on maximization of submodular functions. In particular, algorithms for maximizing a submodular function subject to various constraints have found many applications in machine learning and data mining, including data summarization \cite{mirzasoleiman16distributed, wei2013using}, document summarization \cite{LB10,LB11}, sensor placement \cite{KSG08,KG05}, network reconstruction \cite{chen2016submodular, gomez2010inferring}, crowd teaching \cite{singla2014near}, spread of influence \cite{KKT03} and article recommendation \cite{el2011beyond,mirzasoleiman2016fast}.

%

Despite the above mentioned abundance of settings which give raise to submodular functions, it has been observed that there are also many settings inducing functions that are \emph{close} to submodular (in some sense), but not strictly submodular. Unfortunately, algorithms that have been developed for maximization of true submodular functions often fail miserably when given a function which is only close to submodular. This hurdle has motivated the development of algorithms whose guarantee degrades gracefully with the distance of the function from submodularity. In particular, such algorithms have been developed for functions that are: close to submodular under a distance measure known as the supermodular degree~\cite{FI13,FI14,FI17},
close to a submodular function up to a multiplicative factor~\cite{HS16},
 noisy versions of submodular functions under various noise models~\cite{HS17}, almost submodular in the sense that they satisfy the submodularity inequality $ f(A) + f(B) \geq f(A\cup B) + f(A\cap B) $ up to a fixed constant~\cite{bateni2013submodular}
 or belong to a class of functions known as hypergraph-$ r $ valuations which restricts the kinds of interplay between elements that can affect a function's value~\cite{ABDR12}. 

A particularly important class of close to submodular functions, known as $\gamma$-weakly submodular functions (where $\gamma$ is a parameter measuring the distance of the function from being submodular), has received a lot of attention from the machine learning community. Weakly submodular functions were originally introduced by a work of Das and Kempe~\cite{DK11}, which showed that the standard greedy algorithm achieves a good approximation ratio of $1 - e^{-\gamma}$ for the problem of maximizing such functions subject to a cardinality constraint. Further works developed more sophisticated algorithms for the same maximization problem and demonstrated a large repertoire of applications captured by it. For example, Elenberg et al.~\cite{EDFK17} described a streaming algorithm for the above maximization problem and used it to get a faster algorithm for interpreting outputs of neural networks. By relying on previous works which showed that submodular functions can be maximized by faster versions of the standard greedy algorithm that are either stochastic or distributed \cite{MBKVK15,mirzasoleiman2013distributed},  Khanna et al.~\cite{KEDNG17} showed that these faster versions of the greedy algorithm can also be used for maximizing weakly submodular functions. 
%
Hu et al.~\cite{hu2016efficient} proposed an algorithm that achieved anytime linear prediction by sequencing the computation of feature groups, whose theoretical guarantee is established based on weak submodularity (they term it $ \gamma $-approximate submodularity in the paper). Finally, Qian et al.~\cite{qian2016parallel} leveraged weak submodularity in the design of an approach for the parallel Pareto optimization problem for subset selection, which emerged as a powerful approximate solver for the subset selection problem.

To the best of our knowledge, all existing works regarding the maximization of $\gamma$-weakly submodular functions assume a simple cardinality constraint, and thus, cannot be applied to applications which require more involved constraints. In this paper we make a first step towards amending this situation. Specifically, we show that an algorithm called {\RRG} (originally given by~\cite{BFNS14} for submodular maximization) yields through a more involved analysis the first non-trivial approximation ratio for maximizing a $\gamma$-weakly submodular function subject to a general matroid constraint.

The above result has two important implications. First, as explained above, it opens the door to more involved applications. The second implication is related to the fact that {\RRG} can be viewed as a randomized version of the greedy algorithm. This makes it possible to view our analysis of {\RRG} as an evidence that the standard greedy algorithm works most of time. In other words, we expect the greedy algorithm to produce a good approximation ratio on instances that are not specifically engineered to make it perform poorly. To see whether this is indeed the case, we have conducted four sets of experiments. The first set studies the linear regression problem on synthetic data, and the other sets correspond to real-world application scenarios and use real data (the three real-world application scenarios are video summarization, splice site detection and black-box interpretation of images). 
 As expected, our experiments show that {\RRG} and the greedy algorithm have comparable performance on real-world instances. Moreover, it turns out that the greedy algorithm even manages to outperform {\RRG} on many such instances.


\subsection{Preliminaries and Results}

In this section we present the notation and definitions we use in this paper, including the definition of $\gamma$-weak submodularity. We then use these notation and definitions to present our result formally.

We say that a set function $f\colon 2^\cN \to \bR$ over a ground set $\cN$ is monotone if $f(A) \leq f(B)$ for every two sets $A \subseteq B \subseteq \cN$. Furthermore, given two subsets $A, B \subseteq \cN$, we denote by $f(B \mid A)$ the marginal contribution of adding the elements of $B$ to $A$. More formally, $f(B \mid A) = f(A \cup B) - f(A)$. In many cases the subset $B$ in the above definition will be a singleton set $\{u\}$. In these cases we write, for simplicity, $f(u \mid A)$ instead of $f(\{u\} \mid A)$. Additionally, we occasionally use $A + u$ and $A - u$ as shorthands for the union $A \cup \{u\}$ and the expression $A \setminus \{u\}$, respectively.

Using this notation we can now define $\gamma$-weak submodularity as follows (this definition differs slightly from the original definition of $\gamma$-weakly submodular functions by~\cite{DK11}. We discuss this in more detail later in this section).
\begin{definition}
A set function $f\colon 2^\cN \to \bR$ is $\gamma$-weakly submodular for some $\gamma \in (0, 1]$ if
\begin{equation} \label{eq:weak_submodularity}
	\sum_{u \in B} f(u \mid A) \geq \gamma \cdot f(B \mid A)
\end{equation}
for every two sets $A, B \subseteq \cN$.
\end{definition}

Next, we would like to remind the reader of the formal definition of a matroid.\footnote{We assume basic knowledge of matroid theory and matroid related terms, such as ``base'' and ``rank''. The definition of these terms can be found, for example, in \cite[Volume B]{S03}.}
\begin{definition}
Consider a ground set $\cN$ and a non-empty collection $\cI \subseteq 2^\cN$ of subsets of $\cN$. The pair $(\cN, \cI)$ is a matroid if for every two sets $A, B \subseteq \cN$:
\begin{itemize}
	\item If $A \subseteq B$ and $B \in \cI$, then $A \in \cI$.
	\item If $|A| < |B|$ and $B \in \cI$, then there exists an element $u \in B \setminus A$ such that $A \cup \{u\} \in \cI$.
\end{itemize}
Furthermore, the sets in $\cI$ are called the \emph{independent} sets of the matroid.
\end{definition}

Matroids are important because they capture many natural structures. For example, the set of forests of a graph form a matroid known as the graphical matroid of this graph. Consequently, the maximization of various set functions subject to a general matroid constraint has been studied extensively (see, for example,~\cite{CCPV11,FI14,FNS11}). In this paper we are interested in the problem of maximizing a non-negative monotone $\gamma$-weakly submodular function $f\colon 2^\cN \to \nnR$ subject to a matroid $\cM = (\cN, \cI)$ constraint. In other words, we want to find an independent set of the matroid maximizing $f$. Our main result for this problem is given by the following theorem.
\begin{theorem} \label{thm:main_result}
The {\RRG} algorithm of Buchbinder et al.~\cite{BFNS14} has an approximation ratio of at least $(1 + 1/\gamma)^{-2}$ for the problem of maximizing a non-negative monotone $\gamma$-weakly submodular function subject to a matroid constraint.
\end{theorem}

Two remarks about this result are now in place. First, we would like to point out that {\RRG} is quite efficient. In the analysis of such algorithms it is standard practice to assume that the algorithm has access to two oracles: a value oracle that given a set $S \subseteq \cN$ returns the value of the objective function for that set, and an independence oracle that given $S$ determines whether it is independent. Given such oracles, {\RRG} requires only $O(nk)$ queries to each one of them, where $n$ is the size of the ground set $\cN$ and $k$ is the rank of the matroid (\ie, the size of the largest independent set in it). In the rest of this paper we often use $n$ and $k$ to denote their values as defined here.

Our second remark regarding Theorem~\ref{thm:main_result} is related to the definition of $\gamma$-weak submodularity given above. As mentioned, this definition is slightly different from the original definition of $\gamma$-weakly submodular functions by~\cite{DK11}. The original definition was weaker in the sense that it required Inequality~\eqref{eq:weak_submodularity} to hold only for small sets, \ie, sets whose size is at most comparable to the size of the largest possible feasible solution. For the sake of keeping the definition as clean as possible, we dropped this extra complication from our definition of weak submodularity. However, we note that, when one employs algorithms for weak submodular optimization to solve real-world problems, it is often useful to have the weakest possible definition because this makes it more likely for the real-world objective function to fall into the definition. Thus, we would like to point out that Theorem~\ref{thm:main_result} applies even when the objective function only obeys the following weaker definition (with respect to the matroid $\cM$ defining the constraint). Interestingly, this weaker definition is even weaker than the original definition of~\cite{DK11} for weak submodularity.
\begin{definition}
A set function $f\colon 2^\cN \to \bR$ is $(\gamma, \cM)$-restricted weakly submodular for some $\gamma \in (0, 1]$ and matroid $\cM = (\cN, \cI)$ if
\[
	\sum_{u \in B} f(u \mid A) \geq \gamma \cdot f(B \mid A)
\]
for every two sets $A, B \subseteq \cN$ such that $A \cup B \in \cI$.
\end{definition}

\subsection{Additional Related Work}

The study of the maximization of monotone submodular functions subject to a matroid constraint can be traced back to the 1970's. Nemhauser et al.~\cite{NWF78} and Fisher et al.~\cite{FNW78} proved that the standard greedy algorithm achieves approximation ratios of $1 - 1/e \approx 0.632$ and $1/2$ for this problem when the matroid is a uniform matroid and a general matroid, respectively. The approximation ratio for uniform matroid constraints was discovered, at roughly the same time, to be optimal~\cite{NW78}. However, the question regarding the optimality of the $(1/2)$-approximation algorithm for general matroid constraints remained open for many years. A decade ago, this question was finally solved by a celebrated result of C{\u{a}}linescu et al.~\cite{CCPV11} who described a $(1-1/e)$-approximation algorithm for maximizing a monotone submodular function subject to a general matroid constraint. The result of~\cite{CCPV11} proved that exactly the same approximation ratio can be achieved for the maximization of monotone submodular functions subject to uniform and general matroid constraints, which implies that in some sense general matroid constraints are not more difficult than uniform matroid constraints. Nevertheless, there is still a significant gap between the time complexities of the fastest algorithms known for the two types of constraints~\cite{BFS17,MBKVK15}, and closing this gap (or proving that it cannot be done) remains an intriguing open question for future research.

The result of~\cite{CCPV11} has motivated a long series of works on the maximization of non-monotone submodular functions subject to a matroid constraint~\cite{BFNS14,CVZ14,EN16,FNS11,GV11}. The currently best algorithm of this kind achieves an approximation ratio of $0.385$~\cite{BF16} for general matroid constraints, and no better approximation guarantee is known for uniform matroid constraints. On the inapproximability side, it is known that no polynomial time algorithm can achieve approximation ratios better than $0.491$ and $0.478$ for the maximization of non-monotone submodular functions subject to uniform and general matroid constraints, respectively. Hence, for non-monotone submodular functions there is still a large gap between the best known approximation and inapproximability results. This gap is somewhat bridged by a result of~\cite{F17} giving a $0.432$-approximation algorithm for the maximization of a special class of non-monotone submodular functions, known as symmetric submodular functions, subject to a general matroid constraint.
\paragraph{Paper Organization.}

The rest of this paper is organized as follows. In Section~\ref{sec:algorithm} we present {\RRG} and analyze it formally (some of the technical details of the analysis are deferred to Appendix~\ref{app:low_order_term}). Section~\ref{sec:experiment} describes our experimental results. Finally, Section~\ref{sec:conclusion} contains some concluding remarks and points out an interesting direction for future research.
\section{The Algorithm} \label{sec:algorithm}

In this section we present the {\RRG} algorithm of~\cite{BFNS14} and use it to prove Theorem~\ref{thm:main_result}. The pseudocode of this algorithm is given as Algorithm~\ref{alg:ResidualRandomGreedy}. In this pseudocode we use the notation $\cM / S$ to denote the matroid obtained from the input matroid $\cM$ by contracting a set $S$. Informally, Algorithm~\ref{alg:ResidualRandomGreedy} grows a solution $S$ in $k$ rounds, where each round consists of two steps. In the first step, the algorithm assigns to each element a weight which is equal to the marginal contribution of this element to the current solution $S$. Then, in the second step of the round, the algorithm finds a set $M$ of maximum weight among all sets whose union with the current solution $S$ is independent, and adds a uniformly random element from $M$ to $S$.

\begin{algorithm}
\caption{\textsf{Residual Random Greedy for Matroids}$(f, \cM)$} \label{alg:ResidualRandomGreedy}
Initialize: $S_0 \leftarrow \varnothing$.\\
\For{$i$ = $1$ \KwTo $k$}
{
    Let $M_i$ be a base of $\cM / S_{i - 1}$ maximizing $\sum_{u \in M_i} f(u \mid S_{i - 1})$. \label{line:M_construction}\\
    Let $u_i$ be a uniformly random element from $M_i$.\\
    $S_i \leftarrow S_{i - 1} + u_i$.
}
Return $S_k$.
\end{algorithm}

We begin the analysis of Algorithm~\ref{alg:ResidualRandomGreedy} with the following simple observation.
\begin{observation}
Algorithm~\ref{alg:ResidualRandomGreedy} always outputs a feasible set. Moreover, it uses $O(nk)$ value and independence oracle queries.
\end{observation}
\begin{proof}
The first part of the observation follows immediately from the properties of a matroid. Specifically, it follows from the fact that for any independent set $S$ of size less than $k$ there must exist a non-empty set $M$ such that $S \cup M$ is a base, and moreover, for such a set $M$ any subset of $S \cup M$ is independent.

To see why the second part of the observation holds, we observe that the only line of Algorithm~\ref{alg:ResidualRandomGreedy} which requires access to the oracles is Line~\ref{line:M_construction}. This line can be viewed as having two parts. The first part defines a weight $f(u \mid S_{i - 1})$ for every element of $\cN \setminus S$ (which requires $O(n)$ value oracle queries), while the second part finds a maximum weight independent set in $\cM / S_{i - 1}$ subject to these weights (which requires $O(n)$ independence oralce queries when done using the greedy algorithm). Thus, we get that each execution of Line~\ref{line:M_construction} requires $O(n)$ oracle queries, and the observation follows since this line is executed $k$ times.
\end{proof}

In the rest of this section we analyze the approximation ratio of Algorithm~\ref{alg:ResidualRandomGreedy}. Let us denote by $OPT$ an arbitrary optimal solution, \ie, a set which maximizes $f$ among the independent sets of $\cM$. Since $f$ is monotone, we may assume that $OPT$ is a base of $\cM$. We need to construct, for every $0 \leq i \leq k$, a random set $OPT_i$ for which $S_i \cup OPT_i$ is a base. For the construction we use the following lemma from \cite{B69}, which can be found (with a different notation) as Corollary~39.12a in \cite{S03}.
\begin{lemma} \label{le:perfect_matching_two_bases}
If $A$ and $B$ are two bases of a matroid $\cM = (\cN, \cI)$, then there exists a one to one function $g : A \setminus B \rightarrow B \setminus A$ such that for every $u \in A \setminus B$, $(B + u) - g(u) \in \cI$.
\end{lemma}
For $i = 0$, we define $OPT_0 = OPT$. For $i > 0$, $OPT_i$ is constructed recursively based on the algorithm's behavior. Assume that $OPT_{i - 1}$ is already constructed, and let $g_i\colon M_i \to OPT_{i - 1}$ be a one to one function mapping every element $u \in M_i$ to an element of $OPT_{i - 1}$ in such a way that $S_{i - 1} \cup OPT_{i - 1} - g_i(u) + u$ is a base. Observe that the existence of such a function follows immediately from Lemma~\ref{le:perfect_matching_two_bases} since both $S_{i-1} \cup OPT_{i-1}$ and $S_{i - 1} \cup M_i$ are bases of $\cM$. We now set $OPT_i = OPT_{i - 1} - g_i(u_i)$, which guarantees that $S_i \cup OPT_i$ is a base, as promised. For this construction to be useful, it is important that the choice of $g_i$ (among the possibly multiple functions obeying the required properties) is done independently of the random choice of $u_i$, which guarantees that $g_i(u_i)$ is a uniformly random sample from $OPT_{i - 1}$.

The next lemma proves a lower bound on the expected values of the sets we have constructed. The proof of this lemma is similar to the proof of Lemma~A.1 of~\cite{EDFK17}. Nevertheless, the current lemma manages to get a tighter bound than what was obtained by~\cite{EDFK17}.

\begin{lemma} \label{lem:value}
For every $0 \leq i \leq k$, $\bE[f(OPT_i)] \geq \left[1 - \left(\frac{i + 1}{k + 1}\right)^\gamma\right] \cdot f(OPT)$.
\end{lemma}
\begin{proof}
We first prove by induction that, for every $0 \leq i \leq k$,
\begin{equation} \label{eq:intermidiate}
	\bE[f(OPT_i)] \geq  [1 - e^{-\gamma \cdot \sum_{j = i + 1}^k j^{-1}}] \cdot f(OPT) \enspace.
\end{equation}
For $i = k$, Inequality~\eqref{eq:intermidiate} follows from the non-negativity of $f$ since
\[
	f(OPT_k)
	\geq
	0
	=
	[1 - e^{-\gamma \cdot 0}] \cdot f(OPT)
	=
	[1 - e^{-\gamma \cdot \sum_{j = k + 1}^k j^{-1}}] \cdot f(OPT)
	\enspace.
\]

Assume now that Inequality~\eqref{eq:intermidiate} holds for some $0 < i + 1 \leq k$, and let us prove it holds also for $i$. Observe that $OPT_i$ is a uniformly random subset of $OPT$ of size $k - i$, and $OPT_{i + 1}$ is a uniformly random subset of $OPT$ of size $k - i - 1$. Thus, we can think of $OPT_i$ as obtained from $OPT_{i + 1}$ by adding a uniformly random element of $OPT \setminus OPT_{i + 1}$. Taking this point of view, we get, for every set $S \subseteq OPT$ of size $k - i - 1$,
\begin{align*}
	\bE[f(OPT_i) \mid OPT_{i + 1} = S&]
	=
	f(S) + \frac{\sum_{u \in OPT \setminus S} f(u \mid S)}{|OPT \setminus S|}
	=
	f(S) + \frac{1}{i + 1} \cdot \sum_{u \in OPT \setminus S} f(u \mid S)\\
	\geq{} &
	f(S) + \frac{\gamma}{i + 1} \cdot f(OPT \setminus S \mid S)
	=
	\left(1 - \frac{\gamma}{i + 1}\right) \cdot f(S) + \frac{\gamma}{i + 1} \cdot f(OPT)
	\enspace,
\end{align*}
where the inequality holds by the $\gamma$-weak submodularity of $f$. Taking expectation over the set $OPT_{i + 1}$, the last inequality becomes
\begin{align*}
	\bE[f(OPT_i)]
	\geq{} &
	\left(1 - \frac{\gamma}{i + 1}\right) \cdot \bE[f(OPT_{i + 1})] + \frac{\gamma}{i + 1} \cdot f(OPT)\\
	\geq{} &
	\left(1 - \frac{\gamma}{i + 1}\right) \cdot [1 - e^{-\gamma \cdot \sum_{j = i + 2}^k j^{-1}}] \cdot f(OPT) + \frac{\gamma}{i + 1} \cdot f(OPT)\\
	={} &
	f(OPT) - \left(1 - \frac{\gamma}{i + 1}\right) \cdot e^{-\gamma \cdot \sum_{j = i + 2}^k j^{-1}} \cdot f(OPT)\\
	\geq{} &
	f(OPT) - e^{-\gamma \cdot \sum_{j = i + 1}^k j^{-1}} \cdot f(OPT)
	=
	[1 - e^{-\gamma \cdot \sum_{j = i + 1}^k j^{-1}}] \cdot f(OPT)
	\enspace,
\end{align*}
where the second inequality follows by the induction hypothesis, and the last inequality follows by the inequality $1 - x \leq e^{-x}$ (which holds for every $x$). This completes the proof of Inequality~\eqref{eq:intermidiate}. To see why the lemma follows from this inequality, we observe that
\[
	\sum_{j = i + 1}^k \frac{1}{j}
	\geq
	\int_{i + 1}^{k + 1} \frac{dx}{x}
	=
	\left.\ln x \right|_{i + 1}^{k + 1}
	=
	\ln\left(\frac{k + 1}{i + 1}\right)
	\enspace,
\]
which implies (by Inequality~\eqref{eq:intermidiate})
\[
	\bE[f(OPT_i)]
	\geq
	[1 - e^{-\gamma \cdot \sum_{j = i + 1}^k j^{-1}}] \cdot f(OPT)
	\geq
	\left[1 - \left(\frac{i + 1}{k + 1}\right)^\gamma\right] \cdot f(OPT)
	\enspace.
	\qedhere
\]
\end{proof}

The next observation gives a lower bound on the increase in $\bE[f(S_i)]$ as a function of $i$. Note that the bound given by this observation uses the sets $\{OPT_i\}_{i = 0}^k$ that we have constructed above.

\begin{observation} \label{obs:S_gain}
For every $1 \leq i \leq k$, $\bE[f(S_i)] \geq \bE[f(S_{i - 1})] + \gamma \cdot \frac{\bE[f(OPT_{i - 1} \cup S_{i - 1})] - \bE[f(S_{i - 1})]}{k - i + 1}$.
\end{observation}
\begin{proof}
Fix $1 \leq i \leq k$, and let $A_{i - 1}$ be an arbitrary event fixing all the random decisions of Algorithm~\ref{alg:ResidualRandomGreedy} up to iteration $i - 1$ (including). All the probabilities, expectations and random quantities in the first part of this proof are implicitly conditioned on $A_{i - 1}$. The $\gamma$-weak submodularity of $f$ implies
\[
	\sum_{u \in OPT_{i - 1}} f(u \mid S_{i - 1}) \geq \gamma \cdot f(OPT_{i - 1} \mid S_{i - 1})
	\enspace.
\]
Since $OPT_{i - 1}$ is one possible candidate to be $M_i$, we get
\[
	\sum_{u \in M_i} f(u \mid S_{i - 1})
	\geq
	\sum_{u \in OPT_{i - 1}} f(u \mid S_{i - 1})
	\geq
	\gamma \cdot f(OPT_{i - 1} \mid S_{i - 1})
	\enspace.
\]

Algorithm~\ref{alg:ResidualRandomGreedy} gets $S_i$ by adding a uniformly random element $u_i \in M_i$ to the set $S_{i - 1}$. Since the size of $M_i$ is $k - i + 1$, this implies
\begin{align*}
	\bE[f(S_i)]
	={} &
	f(S_{i - 1}) + \bE[f(u_i \mid S_{i - 1})]
	=
	f(S_{i - 1}) + \frac{1}{k - i + 1} \cdot \sum_{u \in M_i} f(u \mid S_{i - 1})\\
	\geq{} &
	f(S_{i - 1}) + \frac{\gamma}{k - i + 1} \cdot f(OPT_{i - 1} \mid S_{i - 1})
	=
	f(S_{i - 1}) + \gamma \cdot \frac{f(OPT_{i - 1} \cup S_{i - 1}) - f(S_{i - 1})}{k - i + 1}
	\enspace.
\end{align*}
Recall that the last inequality is implicitly conditioned on the event $A_{i - 1}$. The observation now follows by taking the expectation of both sides of this inequality over all possible such events.
\end{proof}

Combining the last lemma and observation, we get the next corollary.

\begin{corollary} \label{cor:opt_plugged}
For every $1 \leq i \leq k$, $\bE[f(S_i)] \geq \bE[f(S_{i - 1})] + \gamma \cdot \frac{\{1 - [i/(k + 1)]^\gamma\} \cdot f(OPT) - \bE[f(S_{i - 1})]}{k - i + 1}$.
\end{corollary}
\begin{proof}
Note that
\begin{align*}
	\bE[f(S_i)]
	\geq{} &
	\bE[f(S_{i - 1})] + \gamma \cdot \frac{\bE[f(OPT_{i - 1} \cup S_{i - 1})] - \bE[f(S_{i - 1})]}{k - i + 1}\\
	\geq{} &
	\bE[f(S_{i - 1})] + \gamma \cdot \frac{\bE[f(OPT_{i - 1})] - \bE[f(S_{i - 1})]}{k - i + 1}\\
	\geq{} &
	\bE[f(S_{i - 1})] + \gamma \cdot \frac{\{1 - [i/(k + 1)]^\gamma\} \cdot f(OPT) - \bE[f(S_{i - 1})]}{k - i + 1}
	\enspace,
\end{align*}
where the first inequality follows by Observation~\ref{obs:S_gain}, the second inequality holds by the monotonicity of $f$, and the last inequality follows by Lemma~\ref{lem:value}.
\end{proof}

We are now ready to prove an approximation ratio for Algorithm~\ref{alg:ResidualRandomGreedy}. Specifically, the next theorem proves that the approximation ratio of Algorithm~\ref{alg:ResidualRandomGreedy} can be smaller than the approximation ratio guaranteed by Theorem~\ref{thm:main_result} by at most a low order term of $O(k^{-1})$. In Appendix~\ref{app:low_order_term}, we show that this low order term can be dropped, which implies Theorem~\ref{thm:main_result}.
\begin{theorem} \label{thm:main_result_weak}
The approximation ratio of Algorithm~\ref{alg:ResidualRandomGreedy} is at least $(1 + 1/\gamma)^{-2} - O(k^{-1})$.
\end{theorem}
\begin{proof}
Let us first prove by induction that, for every $0 \leq i \leq k$,
\begin{equation} \label{eq:theorem_induction}
	\bE[f(S_i)]
	\geq
	\gamma \cdot \frac{\sum_{j = 1}^i \{1 - [j/(k + 1)]^\gamma\} \cdot f(OPT) - i \cdot \bE[f(S_k)]}{k} \enspace.
\end{equation}
For $i = 0$, Inequality~\eqref{eq:theorem_induction} follows from the non-negativity of $f$ since
\[
	\bE[f(S_0)]
	\geq
	0
	=
	\gamma \cdot \frac{0 \cdot f(OPT) - 0}{k}
	=
	\gamma \cdot \frac{\sum_{j = 1}^0 \{1 - [j/(k + 1)]^\gamma\} \cdot f(OPT) - 0 \cdot \bE[f(S_k)]}{k}
	\enspace.
\]
Assume now that Inequality~\eqref{eq:theorem_induction} holds for some $0 \leq i - 1 < k$, and let us prove that it holds for $i$ as well. There are two cases to consider. If $\{1 - [i/(k + 1)]^\gamma\} \cdot f(OPT) \leq \bE[f(S_k)]$, then the monotonicity of $f$ and the fact that $S_{i - 1}$ is a subset of $S_i$ guarantee together that
\begin{align*}
	\bE[f(S_i)]
	\geq{} &
	\bE[f(S_{i - 1})]
	\geq
	\gamma \cdot \frac{\sum_{j = 1}^{i - 1} \{1 - [j/(k + 1)]^\gamma\} \cdot f(OPT) - (i - 1) \cdot \bE[f(S_k)]}{k}\\
	\geq{} &
	\gamma \cdot \frac{\sum_{j = 1}^i \{1 - [j/(k + 1)]^\gamma\} \cdot f(OPT) - i \cdot \bE[f(S_k)]}{k}
	\enspace,
\end{align*}
where the second inequality follows by the induction hypothesis. Thus, it remains to consider the case that $\{1 - [i/(k + 1)]^\gamma\} \cdot f(OPT) \geq \bE[f(S_k)]$. By Corollary~\ref{cor:opt_plugged}, we get in this case
\begin{align*}
	\bE[f(S_i)] - \bE[f(S_{i - 1})]
	\geq{} &
	\gamma \cdot \frac{\{1 - [i/(k + 1)]^\gamma\} \cdot f(OPT) - \bE[f(S_{i - 1})]}{k - i + 1}\\
	\geq{} &
	\gamma \cdot \frac{\{1 - [i/(k + 1)]^\gamma\} \cdot f(OPT) - \bE[f(S_k)]}{k - i + 1}\\
	\geq{} &
	\gamma \cdot \frac{\{1 - [i/(k + 1)]^\gamma\} \cdot f(OPT) - \bE[f(S_k)]}{k}
	\enspace,
\end{align*}
where the second inequality follows by the monotonicity of $f$ since $S_{i - 1}$ is a subset of $S_k$. Adding the last inequality to the induction hypothesis proves that Inequality~\eqref{eq:theorem_induction} holds for $i$, and thus, completes the proof by induction of Inequality~\eqref{eq:theorem_induction}.

Let us now explain why the theorem follows from Inequality~\eqref{eq:theorem_induction}. Plugging $i = k$ into this inequality yields
\[
	\bE[f(S_k)]
	\geq
	\gamma \cdot \frac{\sum_{j = 1}^k \{1 - [j/(k + 1)]^\gamma\} \cdot f(OPT) - k \cdot \bE[f(S_k)]}{k}
	\enspace,
\]
which implies, by extracting $\bE[f(S_k)]$,
\begin{align*}
	\bE[f(S_k)]
	\geq{} &
	\frac{\gamma \cdot f(OPT)}{(1 + \gamma)k} \cdot \sum_{j = 1}^k \left[1 - \left(\frac{j}{k + 1}\right)^\gamma\right]
	\geq
	\frac{\gamma \cdot f(OPT)}{(1 + \gamma)k} \cdot \int_1^{k + 1} \left[1 - \left(\frac{x}{k + 1}\right)^\gamma\right] dx\\
	={} &
	\frac{\gamma \cdot f(OPT)}{(1 + \gamma)k} \cdot \left[x - \frac{k + 1}{1 + \gamma} \cdot \left(\frac{x}{k + 1}\right)^{1 + \gamma} \right]_1^{k + 1}
	\geq
	\frac{\gamma \cdot f(OPT)}{(1 + \gamma)k} \cdot \left[k - \frac{k + 1}{1 + \gamma} \right]\\
	={} &
	\frac{\gamma \cdot f(OPT)}{(1 + \gamma)k} \cdot \left[\frac{\gamma k}{1 + \gamma} - \frac{1}{1 + \gamma} \right]
	\geq
	\left[\left(\frac{\gamma}{1 + \gamma}\right)^2 - \frac{1}{k} \right] \cdot f(OPT)
	\enspace.
	\qedhere
\end{align*}
\end{proof}

\section{Experiments}

\label{sec:experiment}

We have conducted four sets of experiments. The first set studies linear regression on synthetic data (\cref{sub:linear_regression}), and the other three sets correspond to real-world application scenarios and use real data. The application scenarios studied by these sets are video summarization (\cref{sub:vidsum}), splice site detection (\cref{sub:splice}) and black-box interpretation for images (\cref{sub:interp}). 

\subsection{Linear Regression} \label{sub:linear_regression}


In this set of experiments, we are given an $ n\times p $ matrix $\bm{X}$ and a vector $\bm{y} \in \bR^n$ which is a noisy version of the product of the matrix $\bm{X}$ and an unknown vector $\bm{\beta} \in \bR^p$. More formally, $\bm{y} \triangleq \bm{X} \bm{\beta} + \bm{\eps} $, where the coefficients of noise vector $\bm{\eps}$ are i.i.d.\ standard Gaussian random variables. In general, people are interested in the problem of, given such a matrix $\bm{X}$ and a vector $\bm{y}$, recovering $\beta$ under the assumption that it is sparse in some sense. In the current set of experiment, we call a vector $\bm{\beta}$ sparse if and only if its support $\supp(\bm{\beta})$ is independent in some input matroid $M$. In other words, we want to find among the vectors whose support is independent in $M$, the vector $\bm{\beta}$ which is the most likely to be the vector which has been used to generate $\bm{y}$.

The log-likelihood function of this problem for vectors is given by
\[ l(\bm{\beta}) = - \lVert \bm{y} - \bm{X\beta} \rVert^2 +C \enspace, \]
where $ C $ is a constant, and this yields the following log-likelihood function for support vectors
\[ g(S) = \max_{\supp(\bm{\beta}) \subseteq S } l(\bm{\beta}) = - \lVert  \bm{y} - \bm{X}_S (\bm{X}_S^T \bm{X}_S)^{-1} \bm{X}_S^T \bm{y} \rVert^2 +C \enspace,\]
which was shown to be weakly submodular by~\cite{elenberg2016restricted}.
Thus, our objective is to find a set $S$ which is independent in $M$ and approximately maximizes this weakly submodular function (given such a set $S$, one calculate the vector $\bm{\beta}$ that we look for). Towards this goal, we have applied {\RRG} to the matroid $M$ and the normalized log-likelihood function $ f(S) \triangleq g(S) - g(\varnothing) $ (we do not apply \Alg directly to the log-likelihood function $g$ since the last function is not guaranteed to be non-negative). We then compare the performance of \Alg on this optimization problem with the following baselines.
\begin{itemize}
	\item \AlgRand. The \AlgRand algorithm samples an independent set in an iterative manner. Throughout its execution, \AlgRand maintains an independent set $ S $, which is originally initialized to be the empty set. In each iteration, {\AlgRand} adds to $S$ a uniformly random element from the set of elements in $\cN$ whose addition to $S$ keeps $S$ independent. This process continues until $S$ becomes a base, at which point no more elements can be added to $S$. 
	\item \AlgGreedy. Like the previous algorithm, the \AlgGreedy algorithm maintains an independent set $ S $, which is originally initialized to be the empty set and grows iteratively. In each iteration, {\AlgGreedy} adds to $S$ the element with the largest marginal contribution (with respect to the objective function) among the elements of $\cN$ whose addition to $S$ keeps $S$ independent. Once more, the process terminates when $S$ becomes a base.
\end{itemize}

Before we present the results obtained by this set of experiments, we would like to explain the way we used to generate the inputs for the experiments. We chose $n = 100$ and $p = 200$, and constructed each row of the $ n\times p $ matrix $ \bm{X} $ independently according to an autoregressive (AR) process with $ \alpha = 0.5 $ and noise variance $ \sigma^2 = 10 $ (in this generation process, each entry of the row is a function of the last few entries appearing before it in the row and a few random bits). The support of the vector $\bm{\beta}$ was chosen randomly using the above mentioned algorithm $\AlgRand$ (notice that this algorithm does not use the objective function, and thus, can be understood as a way to sample an independent set from a matroid), and each non-zero value of $\bm{\beta}$ was assigned a uniformly random value from the set $\{-1, 1\}$. It remains to explain the way we have used to construct the matroid $M$ itself, which differs between two experiments we conducted.

In one experiment we have used graphic matroid $M$. To generate the graph underlying this graphic matroid, we started with an empty graph over $n$ vertices and added to it $p$ random edges, where each edge was chosen independently and connected a uniformly random pair of distinct vertices. In the other experiment we have used a partition matroid $M$ with $10$ partitions, which we denote by $B_1, B_2, \dotsc, B_{10}$. To generate this partition matroid we used a few steps. First, we uniformly sampled a random distribution out of the set of all possible distributions over the $10$ partitions (\ie, the sampled distribution is a uniformly random point from the standard 9-simplex).
Then, we created $p$ elements, and assigned each one of them to 
one of the partitions according to the above mentioned distribution.
 Finally, for every partition $B_i$, we sampled its capacity---\ie, the maximum number of elements of this partition that can appear in an independent set---from the binomial distribution $ B(|B_i|, 0.25) $. 

The results of the two experiments are illustrated in~\cref{fig:linear_regression}. The plots in this figure show how the normalized log-likelihood varies as the algorithms select more elements (also called ``features'') under the constraints corresponding to the above matroids. In both plots, the black line denotes the normalized log-likelihood achieved by the ground truth (\ie, the vector $\bm{\beta}$ used to generate $\bm{y}$). We observe that \Alg and \AlgGreedy yield comparable performance and both outperform \AlgRand. In particular, when they terminate, the normalized likelihoods attained by \Alg and \AlgGreedy are almost equal.

\begin{figure*}[htb]
	\centering
	\subfigure[Under a graphic matroid constraint  \label{fig:linear_matroid}]{
		\includegraphics[width=0.47\linewidth]{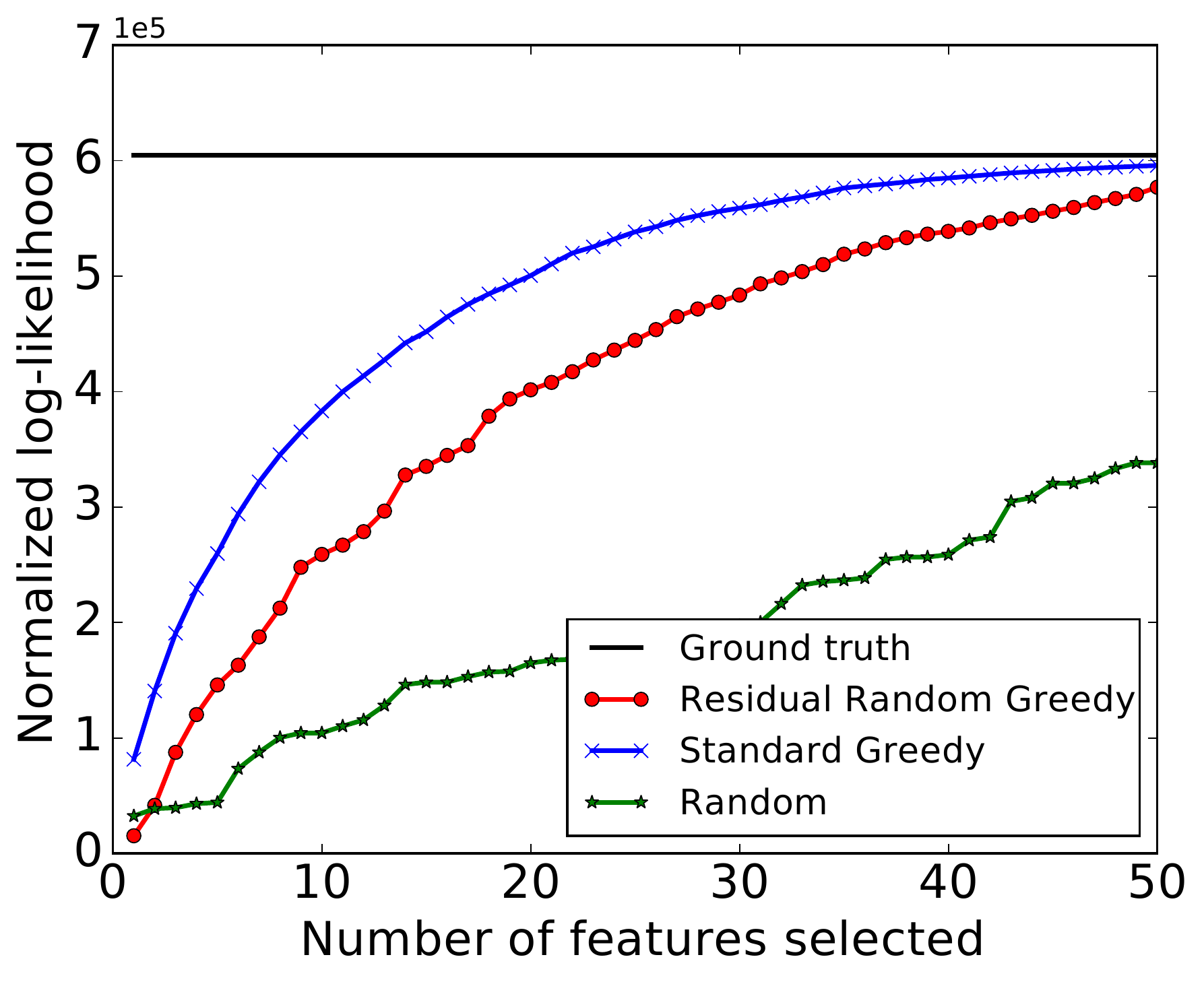}
	}\subfigure[Under a partition matroid constraint \label{fig:partition_matroid}]{
		\includegraphics[width=0.49\linewidth]{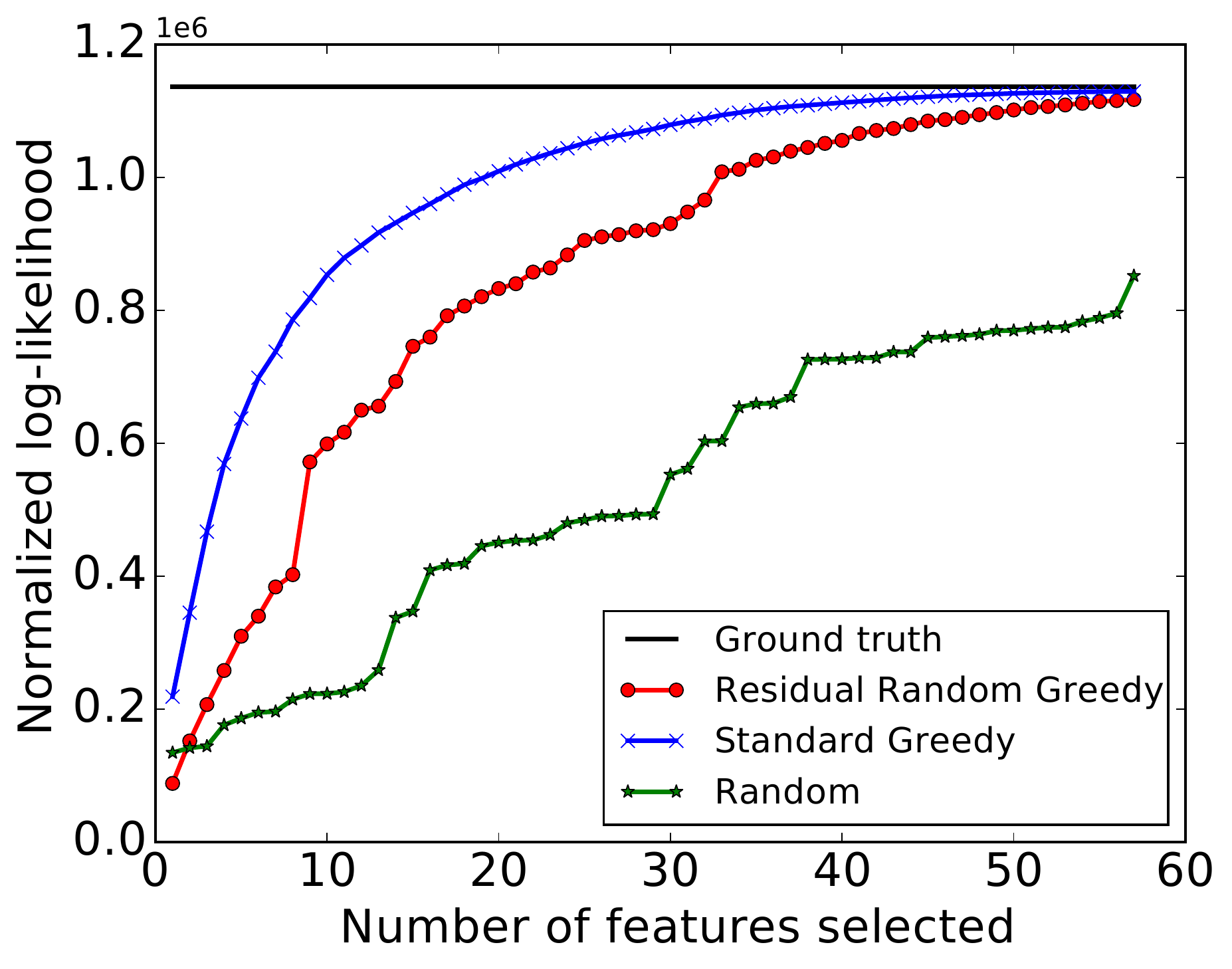}
	}
	%
	%
	\caption{Normalized log-likelihood vs.\ the number of features selected in linear regression. } \label{fig:linear_regression}
\end{figure*}

\subsection{Video Summarization}
\label{sub:vidsum}

In this application our objective is to pick a few frames from a video which summarize it (in some sense). One can formalize the problem of selecting such a summary as selecting a set of frames maximizing the Determinantal Point Process (DPP) objective function, which is a computationally efficient tool that favors subsets of a ground set of items with higher diversity~\cite{kulesza2012determinantal}. More formally, given an $ n $ frames video, we have represented each frame by a $ p $-dimensional  vector. Let $ X\in \mathbb{R}^{n\times n} $ be the Gramian matrix of the $ n $ resulting vectors and the Gaussian kernel; \ie, $ X_{ij} $ is the value of the Gaussian kernel between the $ i $-th and $ j $-th vectors. The DPP objective function is now given as the determinant function $ f:2^{[n]} \to \mathbb{R} $:
\[ 
f(S) = \det (I + X_S ) \enspace,
\]
where $ X_S $ is the principal submatrix of $ X $ indexed by $ S $. 
We note that the identity matrix was added here to the objective to make sure that the function $ f $ is monotone. Moreover, this function was shown to be weakly submodular by~\cite{bian2017guarantees}, although it might not be submodular. In light of the non-submodularity of the determinant function $f$, rather than optimize it directly, prior works considered its log, which is known to be submodular~\cite{kulesza2012determinantal,xu2015gaze}. This allows the use of standard submodular optimization techniques, but does not guarantee any approximation ratio for the original objective function. 
Fortunately, with the help of \Alg, we can maximize the determinant function $f$ directly and get a guranteed approximation ratio.

The video that we have selected for this experiment lasts for 
roughly $7$ minutes and a half, and we chose to created a summary of it
by extracting one representative frame from every 25 seconds. In other words, the constraint on the allowed summerization is given by a partition matroid in which a set $S$ of frames is independent (\ie, belongs to $ \mathcal{I} $) if and only if 
\[ 
|S\cap [25(i-1)+1, 25i] | \leq 1 \qquad \forall\; 1\leq i \leq \lceil n/25 \rceil \enspace.
\]
Given this constraint, the optimization problem that we need to solve is
\[ 
\max_{S\in \mathcal{I}} f(S) \enspace.
\]
\begin{figure}[htb]
	\centering

	\includegraphics[width=0.5\textwidth]{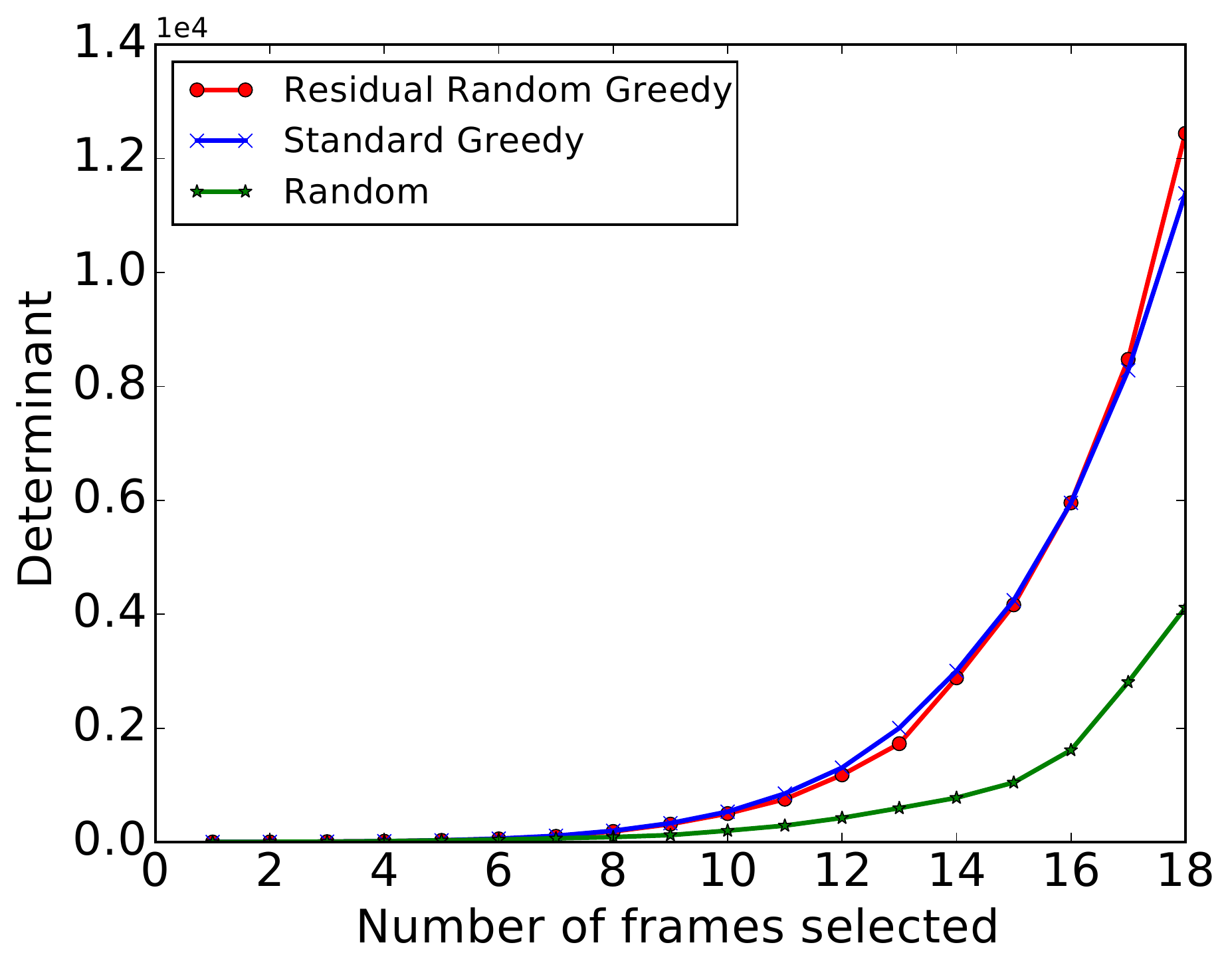}
	\caption{Determinant vs.\ the number of frames selected in the video summarization problem. \label{fig:video_summarization} }
\end{figure}%
\Cref{fig:video_summarization} illustrates the performance of {\RRG} and the two benchmark algorithms when they are applied to this problem. 
The frames selected by the three algorithms are shown in~\cref{fig:vidsum_random_greedy,fig:vidsum_greedy,fig:vidsum_random}. Each algorithm selects one frame per 25 seconds, and the selected 18 frames are arranged in these images in chronological order from left to right and from top to bottom.
It is quite easy to observe that both \Alg and \AlgGreedy produce summaries with higher diversity than \AlgRand. For example, the first two frames selected by \AlgRand are about the same young lady in red, while \Alg and \AlgGreedy choose one about the young lady and the other one about the TV show studio; and again, the 12-th and 13-th frames selected by \AlgRand are both about a lady in black, while \Alg and \AlgGreedy do not produce duplications, which allows them to cover other content. Comparing the outputs of \Alg and \AlgGreedy is more subtle, but the result of {\RRG} seems to be slightly better. The 10-th and 14-th frames selected by \Alg show two participants that are not recognized by the other two summaries; in contrast, \AlgGreedy chooses five frames about TV show guests sitting behind a long blue desk in the studio, which reduces the diversity of the frames.

\begin{figure*}[t]
	\begin{minipage}{\linewidth}
		\centering
		\includegraphics[width=0.105\linewidth]{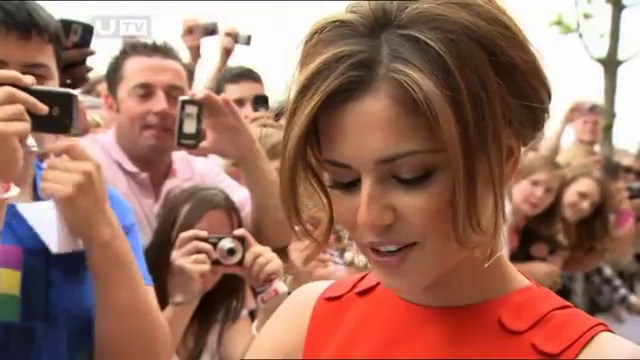}
		\includegraphics[width=0.105\linewidth]{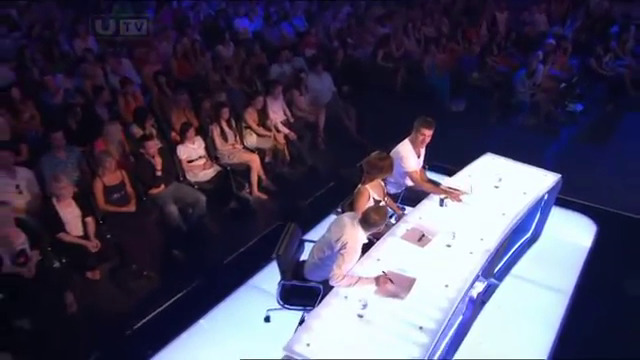}
		\includegraphics[width=0.105\linewidth]{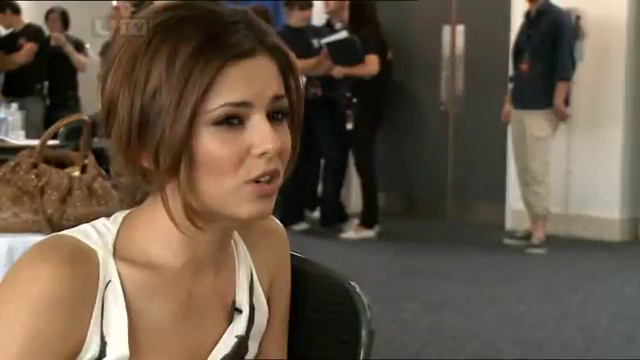}
		\includegraphics[width=0.105\linewidth]{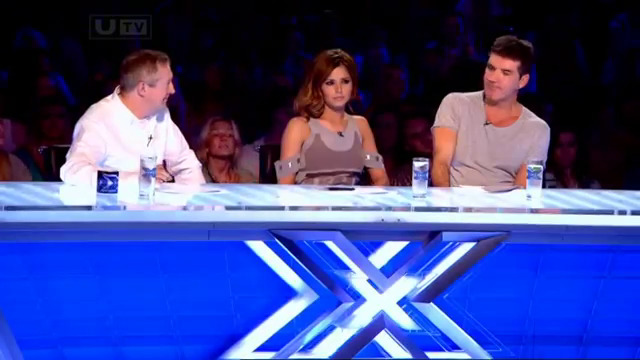}
		\includegraphics[width=0.105\linewidth]{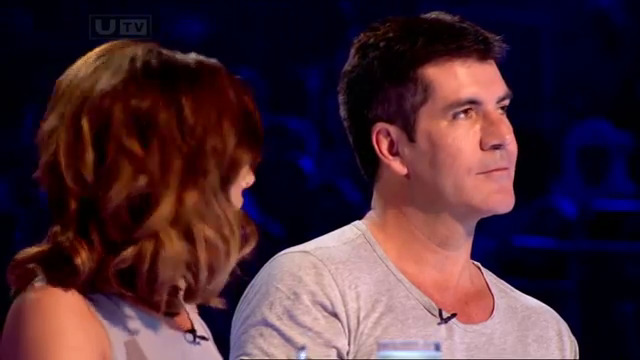}
		\includegraphics[width=0.105\linewidth]{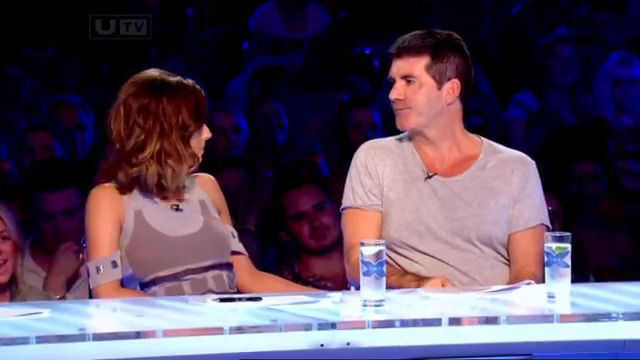}
		\includegraphics[width=0.105\linewidth]{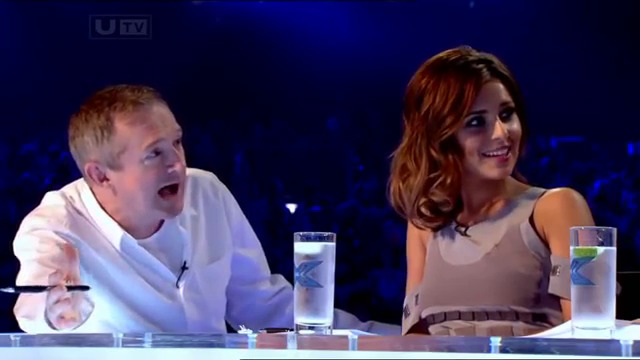}
		\includegraphics[width=0.105\linewidth]{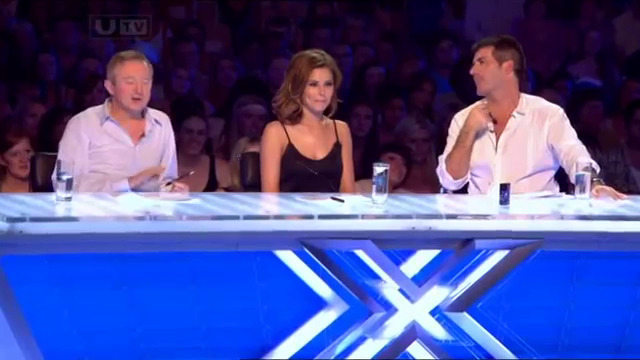}
		\includegraphics[width=0.105\linewidth]{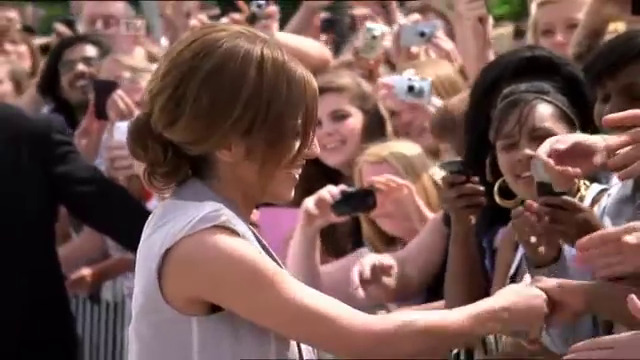}
		
		\smallskip
		\includegraphics[width=0.105\linewidth]{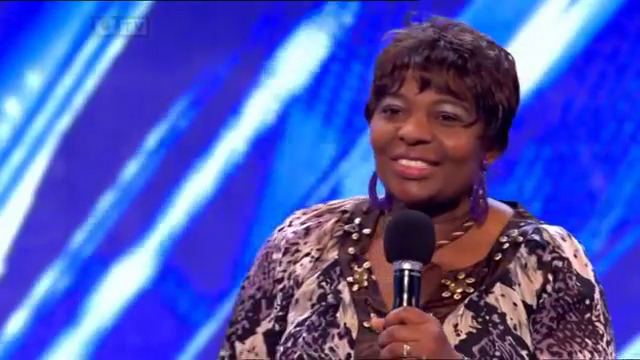}
		\includegraphics[width=0.105\linewidth]{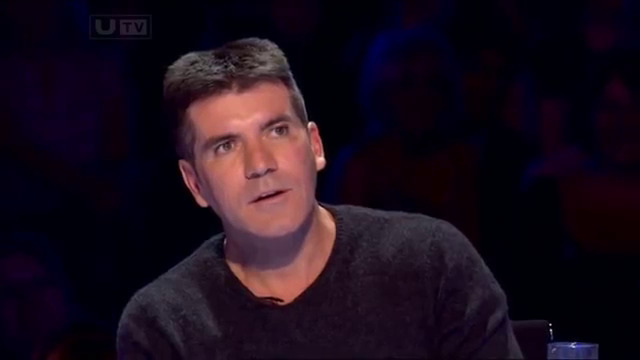}
		\includegraphics[width=0.105\linewidth]{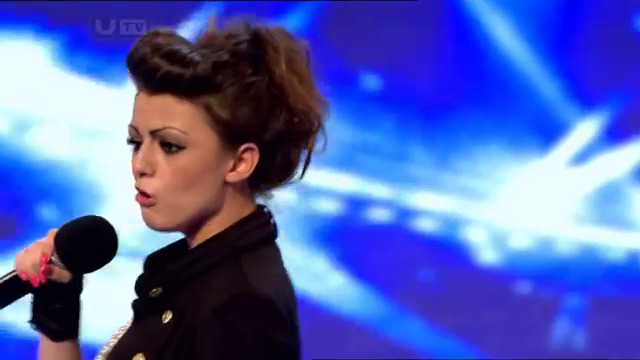}
		\includegraphics[width=0.105\linewidth]{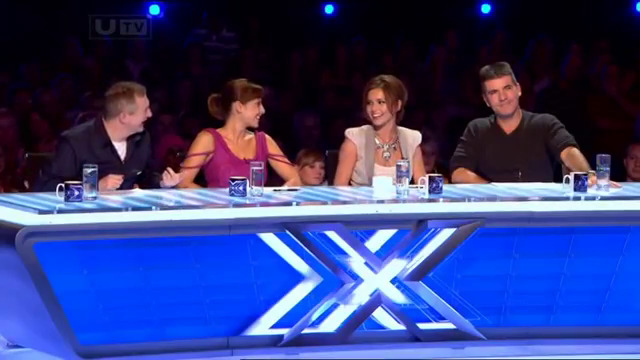}
		\includegraphics[width=0.105\linewidth]{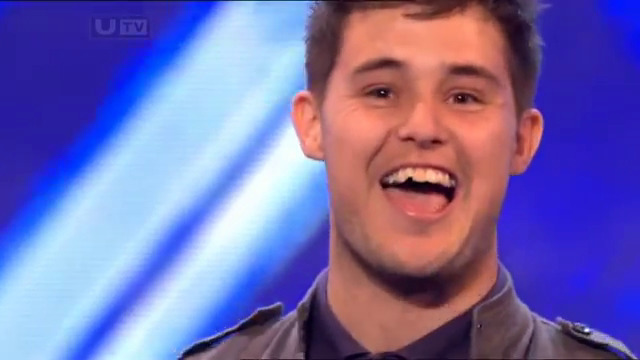}
		\includegraphics[width=0.105\linewidth]{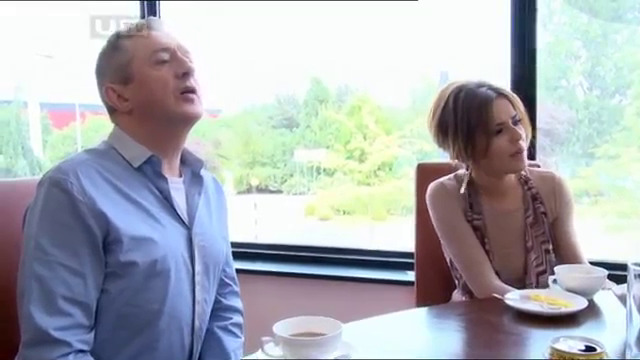}
		\includegraphics[width=0.105\linewidth]{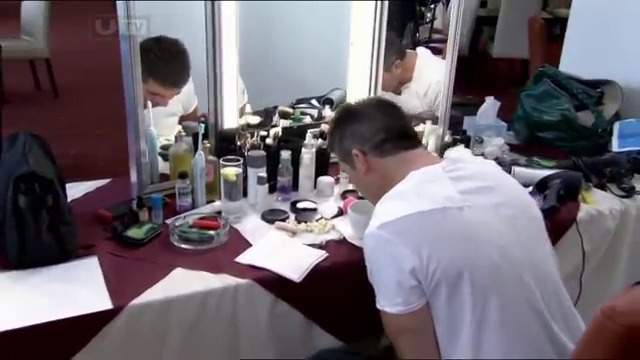}
		\includegraphics[width=0.105\linewidth]{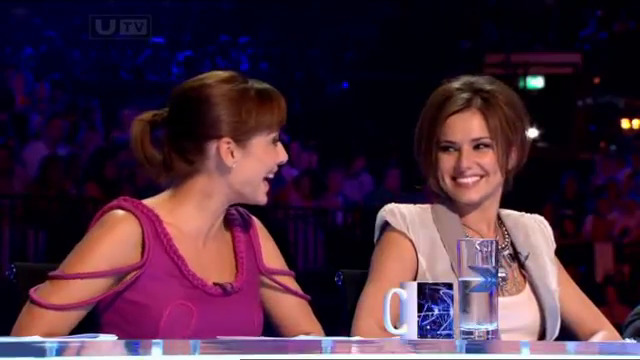}
		\includegraphics[width=0.105\linewidth]{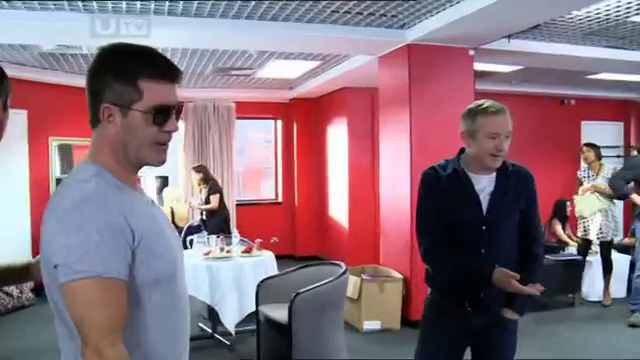}
		\caption{Frames selected by \Alg as its video summary. \label{fig:vidsum_random_greedy} }.  
	\end{minipage}
	\begin{minipage}{\linewidth}
		\centering
		\includegraphics[width=0.105\linewidth]{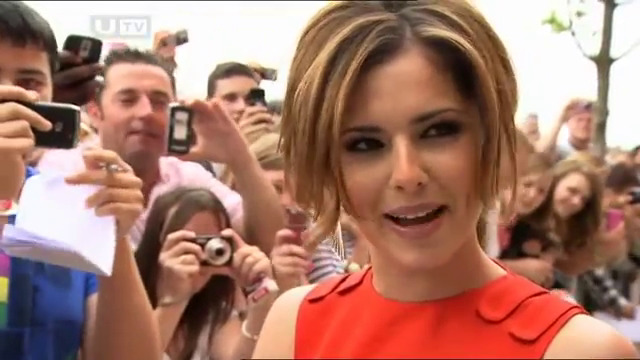}
		\includegraphics[width=0.105\linewidth]{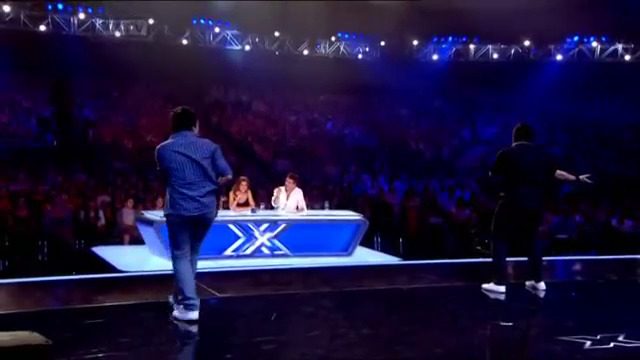}
		\includegraphics[width=0.105\linewidth]{pix/frame/F00075.jpg}
		\includegraphics[width=0.105\linewidth]{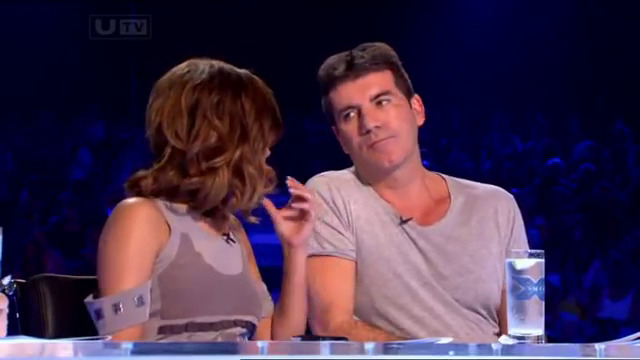}
		\includegraphics[width=0.105\linewidth]{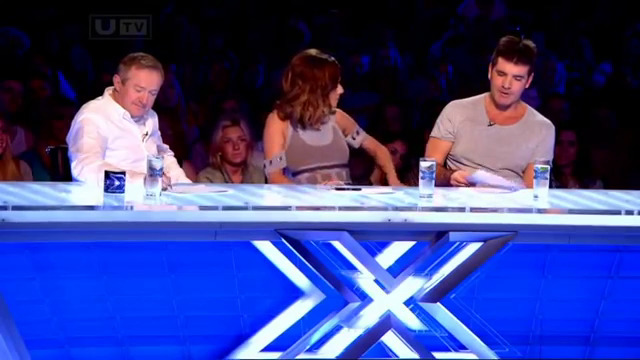}
		\includegraphics[width=0.105\linewidth]{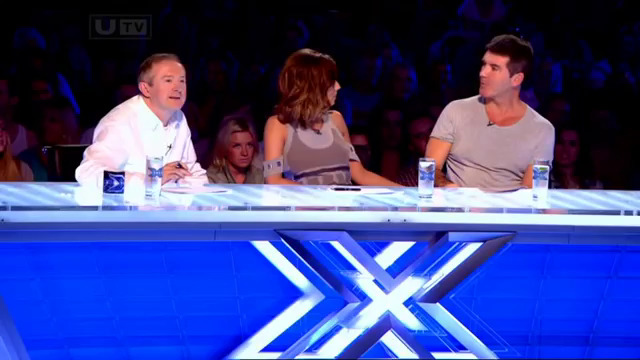}
		\includegraphics[width=0.105\linewidth]{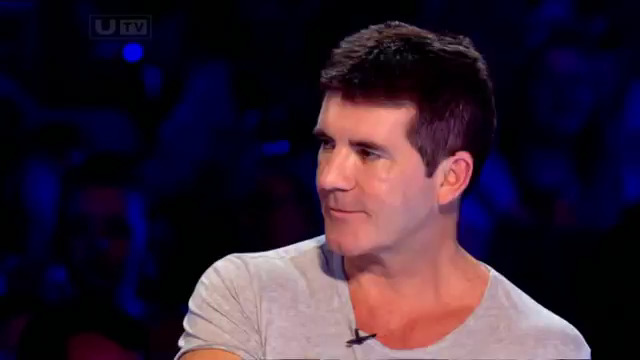}
		\includegraphics[width=0.105\linewidth]{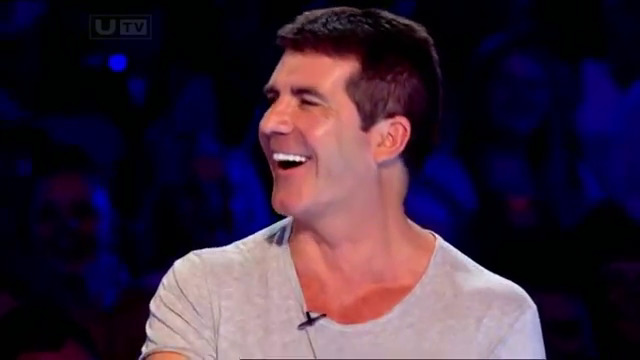}
		\includegraphics[width=0.105\linewidth]{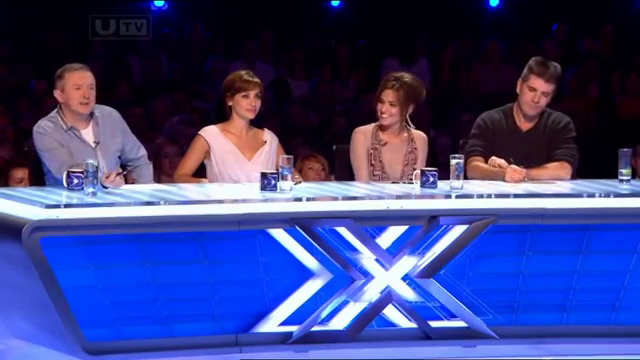}
		
		\smallskip
		\includegraphics[width=0.105\linewidth]{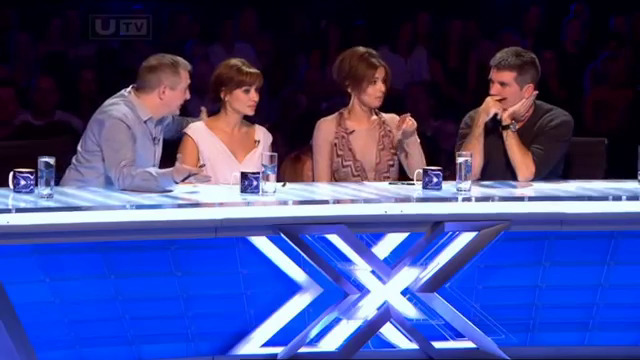}
		\includegraphics[width=0.105\linewidth]{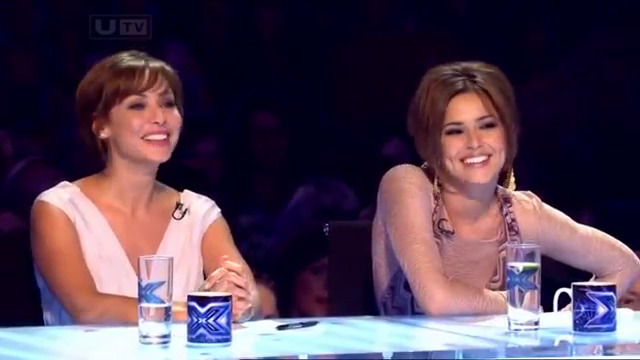}
		\includegraphics[width=0.105\linewidth]{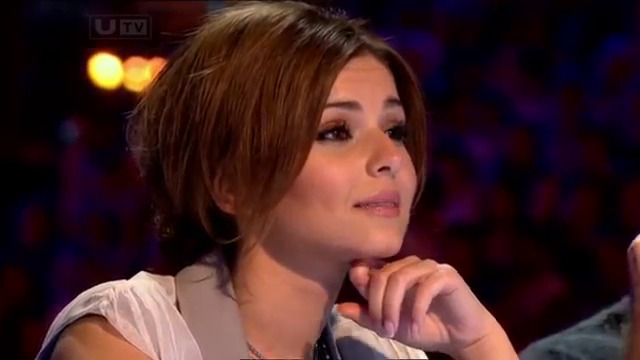}
		\includegraphics[width=0.105\linewidth]{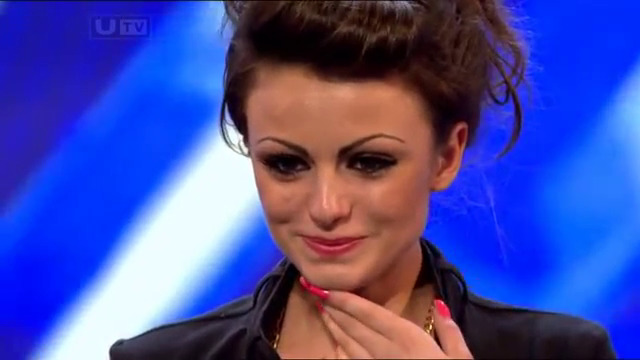}
		\includegraphics[width=0.105\linewidth]{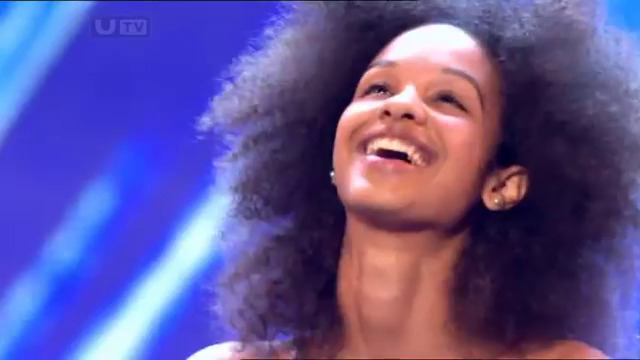}
		\includegraphics[width=0.105\linewidth]{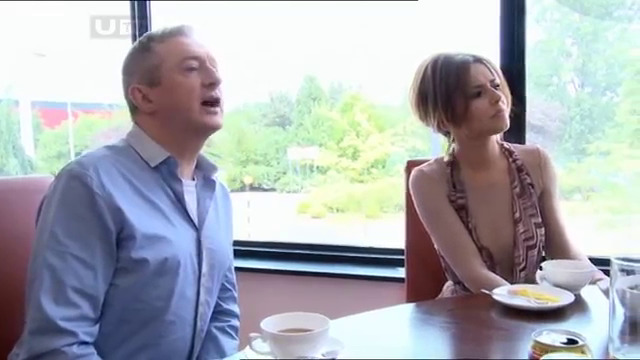}
		\includegraphics[width=0.105\linewidth]{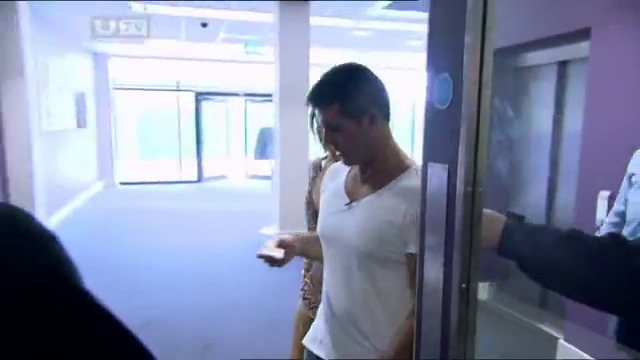}
		\includegraphics[width=0.105\linewidth]{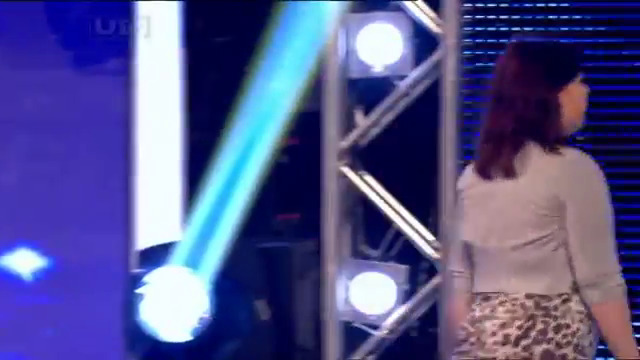}
		\includegraphics[width=0.105\linewidth]{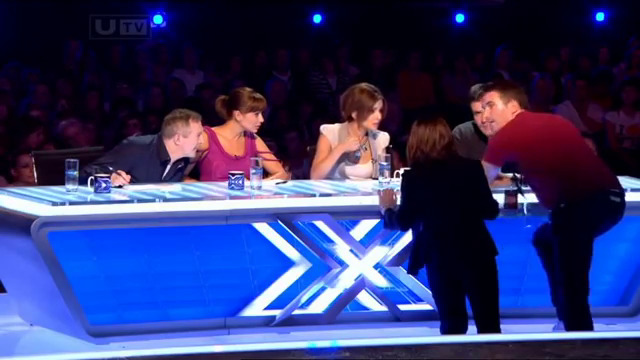}
		\caption{Frames selected by \AlgGreedy as its video summary. \label{fig:vidsum_greedy} }.  
	\end{minipage}
	\begin{minipage}{\linewidth}
		\centering
		\includegraphics[width=0.105\linewidth]{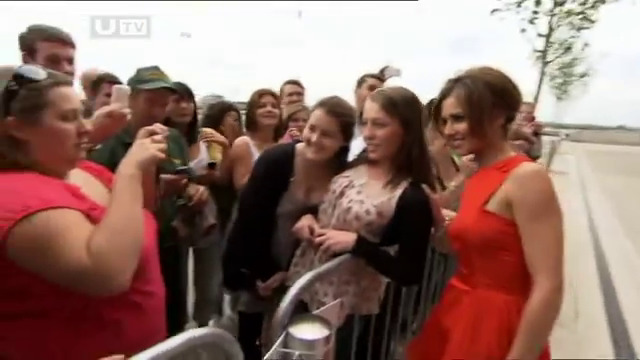}
		\includegraphics[width=0.105\linewidth]{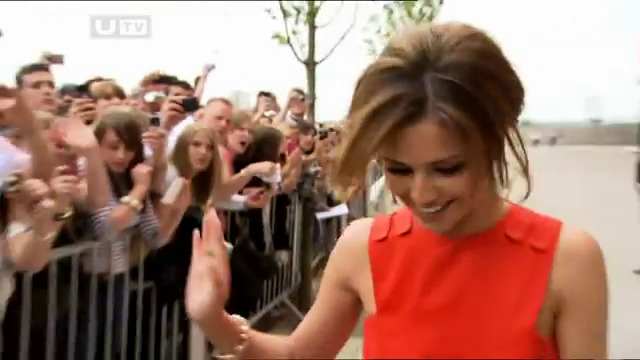}
		\includegraphics[width=0.105\linewidth]{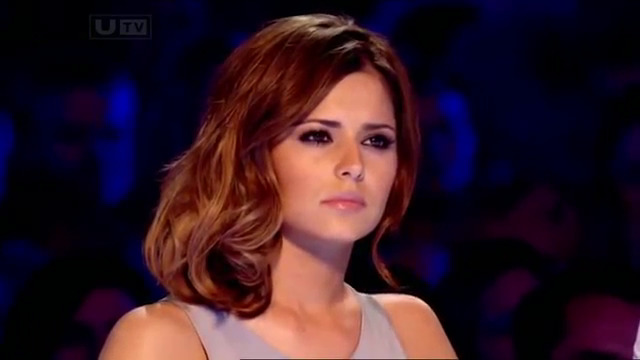}
		\includegraphics[width=0.105\linewidth]{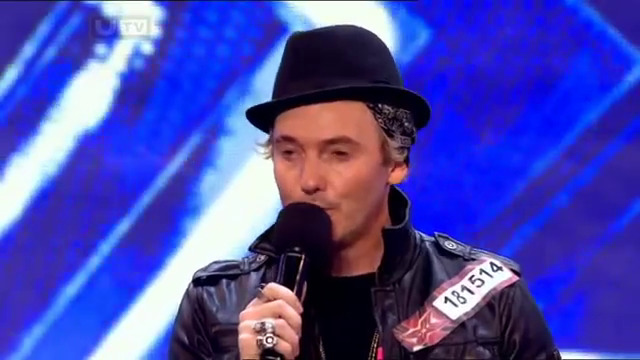}
		\includegraphics[width=0.105\linewidth]{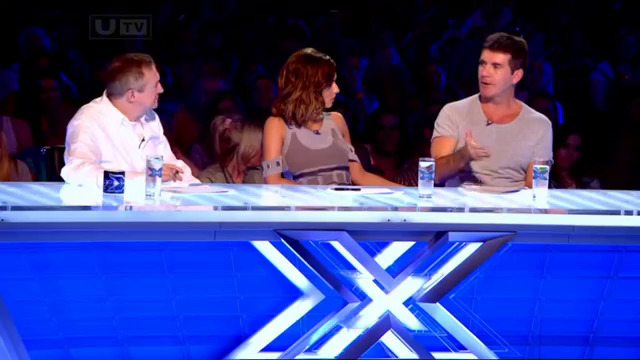}
		\includegraphics[width=0.105\linewidth]{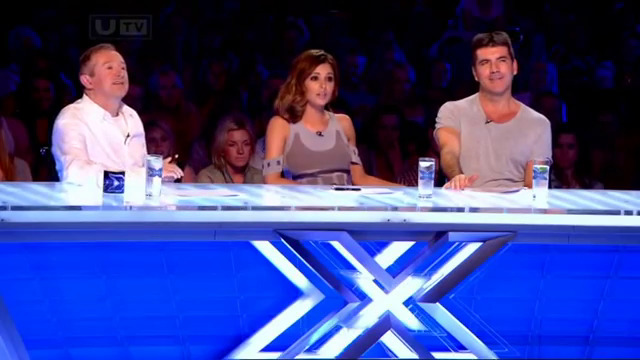}
		\includegraphics[width=0.105\linewidth]{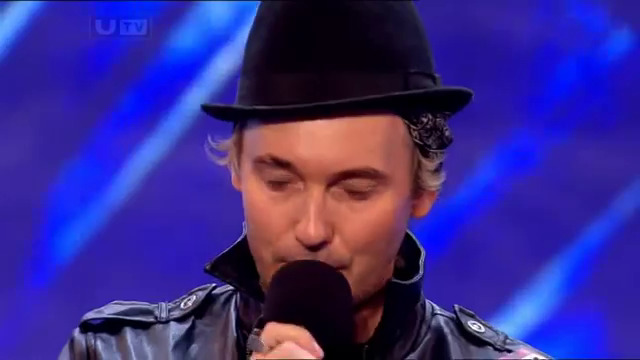}
		\includegraphics[width=0.105\linewidth]{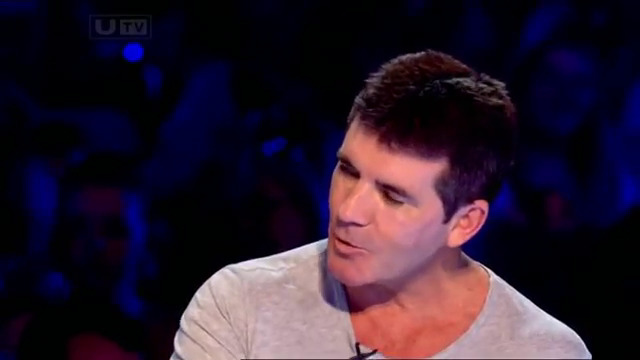}
		\includegraphics[width=0.105\linewidth]{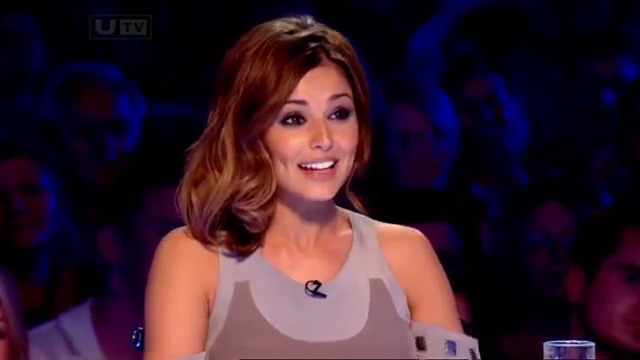}
		
		\smallskip
		
		\includegraphics[width=0.105\linewidth]{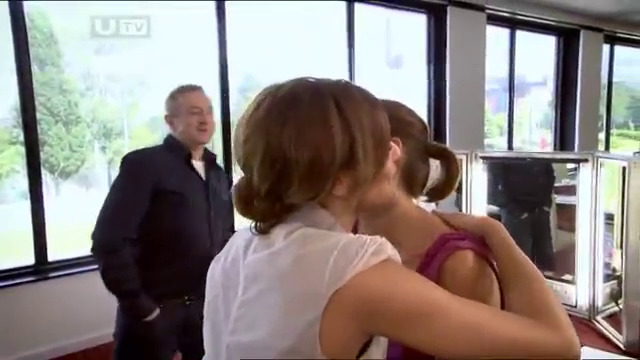}
		\includegraphics[width=0.105\linewidth]{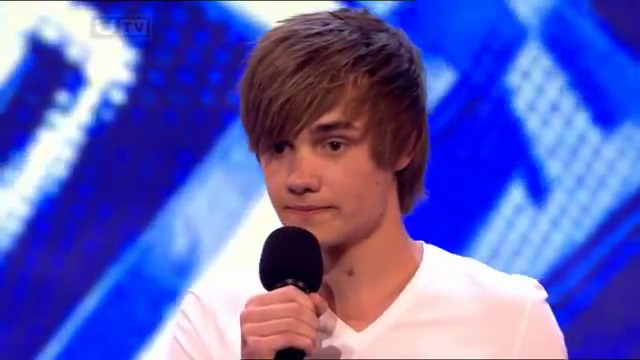}
		\includegraphics[width=0.105\linewidth]{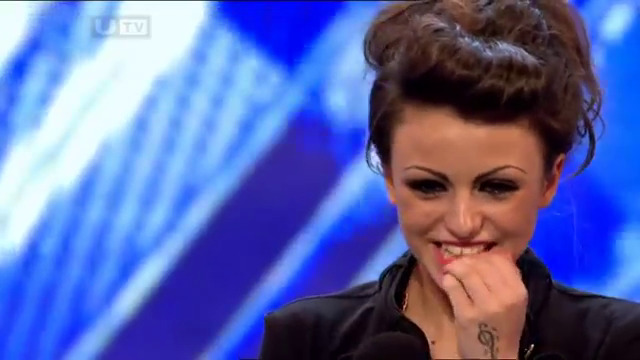}
		\includegraphics[width=0.105\linewidth]{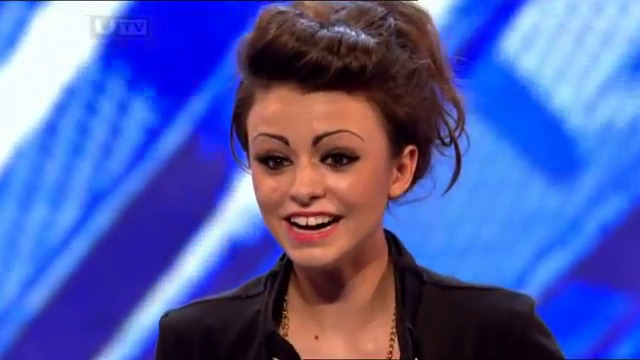}
		\includegraphics[width=0.105\linewidth]{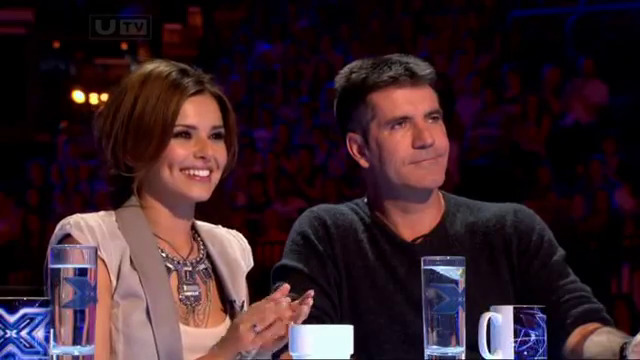}
		\includegraphics[width=0.105\linewidth]{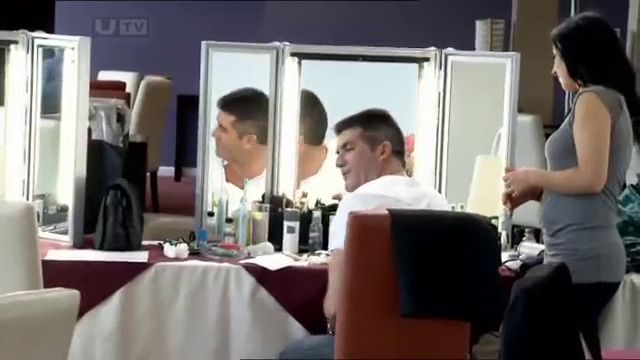}
		\includegraphics[width=0.105\linewidth]{pix/frame/F00401.jpg}
		\includegraphics[width=0.105\linewidth]{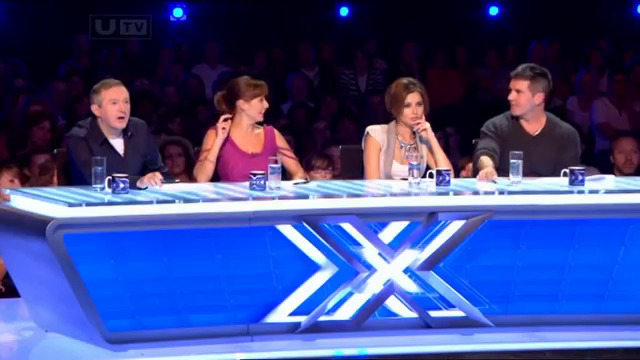}
		\includegraphics[width=0.105\linewidth]{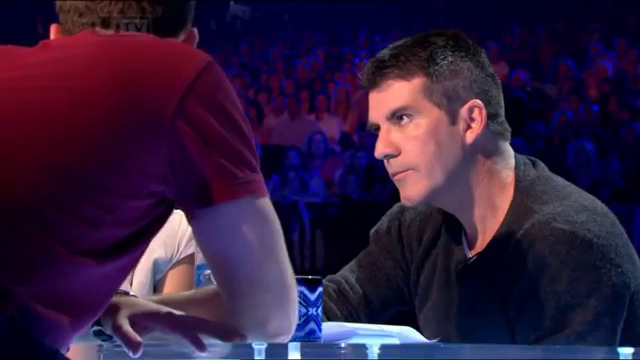}
		\caption{Frames selected by \AlgRand as its video summary. \label{fig:vidsum_random} }.  
	\end{minipage}
\end{figure*}

\subsection{Splice Site Detection}\label{sub:splice}

An important problem in computational biology is the identification of true splice sites from similar decoy splice sites in 
nascent precursor messenger RNA (pre-mRNA) transcripts.
Splice sites are nucleotide sequences that mark the beginnings and ends of introns 
(nucleotide sequences removed by RNA splicing during maturation of mRNA).
In general, the two ends of an RNA sequence are known the 5'-end and the 3'-end. In the case of introns, these ends are also known as the splice donor site and the splice acceptor site, respectively. We are interested in the problem of identifying splice donor and splice acceptor sites. In other words, given a sequence of nucleotides, we want to determine whether this sequence represents a splice donor/acceptor site. A splice donor site always includes the nucleotide sequence ``GT'' at its 5'-end, while a splice acceptor site has the sequence ``AG'' at its 3'-end. However, both kinds of sites include additional nucleotides whose identify should be taken into account when deciding whether a given sequence of nucleotides is a splice donor/acceptor site 

The MEMset dateset provides instances of true and false splice donor/acceptor sites. We note that false splice donor/acceptor sites also include the compulsory ``GT''/``AG'' sequences, but differ from true sites in their other nucleotides. A detailed description of this dataset is presented in~\cite{yeo2004maximum}. In this set of experiments, we used logistic regression on the MEMset dateset to determine the nucleotide values that have the largest influence on the categorization of splice sites into true and false sites. As a preprocessing step, we removed the consensus ``GT'' and ``AG'' sequences. Then, we considered the natural explanatory variables for this problem, \ie, a single variable taking the four values A, C, T and G for every nucleotide of the splice site. As these explanatory variables are categorical; we converted each of them into four binary variables via one-hot encoding. 
In other words, for each explanatory variable $ x_i $ (which takes values from $ \{\text{A, C, T, G}\} $), we created four binary dummy variables $ x'_{4i-3},x'_{4i-2},x'_{4i-1},x'_{4i} $, where 
$ x'_{4i-3}  $ ($ x'_{4i-2}$,$x'_{4i-1}$ and $x'_{4i} $) takes the value one exactly when $ x_i $ is  A  (C, T and G, respectively).
Given this encoding, a natural constraint is that at most one of the four binary variables can be set to one; which is a partition matroid constraint. 
Let us denote the $ j $-th set of binary dummy variables and the corresponding outcome variable by $ x'_{i, j} $ and $ y_j $, respectively.
As is standard in logistic regression, 
we assume that for all $ j $,
\[ \log \left( \frac{p_j}{1-p_j}  \right) = \sum_i w_i x'_{i, j} \]
and \[ y_j \mid \{ x'_{i, j} : 1\leq i\leq 4n \} \sim \mathrm{Bernoulli}(p_j) \enspace, \]
where $ n $ is the total number of categorical explanatory variables (thus, we have $ 4n $ binary dummy variables in total). 
The log-likelihood function of logistic regression can now be written as
\[ l(w) = \sum_j y_j(\sum_i w_i x'_{i, j}) - \log(1+ \exp(\sum_i w_i x'_{i, j}) ) \enspace. \]
As mentioned above, our objective is to find the set of nucleotide values that has the most influence of this log-likelihood function. Thus, the objective function we want to optimize is the normalized log-likelihood
$ f(S) \triangleq g(S)-g(\varnothing) $, where 
\[ g(S) = \max_{ w:  \supp(w)\subseteq S } \mspace{-9mu}  l(w) \enspace. \]
The weak submodularity of this objective function was shown in~\cite{elenberg2016restricted}. In \cref{fig:sites}, we present the result of applying {\RRG} and the two benchmark algorithms mentioned above to this optimization problem. The ranks of the partition matroids for the donor and acceptor sites in \cref{fig:sites} are $ 7 $ and $ 21 $, respectively, because this is the number of nucleotides provided for each one of these kinds of sites by the MEMset dataset. One can note that {\AlgGreedy} and {\RRG} exhibit comparable performance (especially at their termination point) and consistently outperform {\AlgRand}. 

\begin{figure*}[htb]
	\centering
	\subfigure[Splice donor site detection \label{fig:donor}]{
		\includegraphics[width=0.49\linewidth]{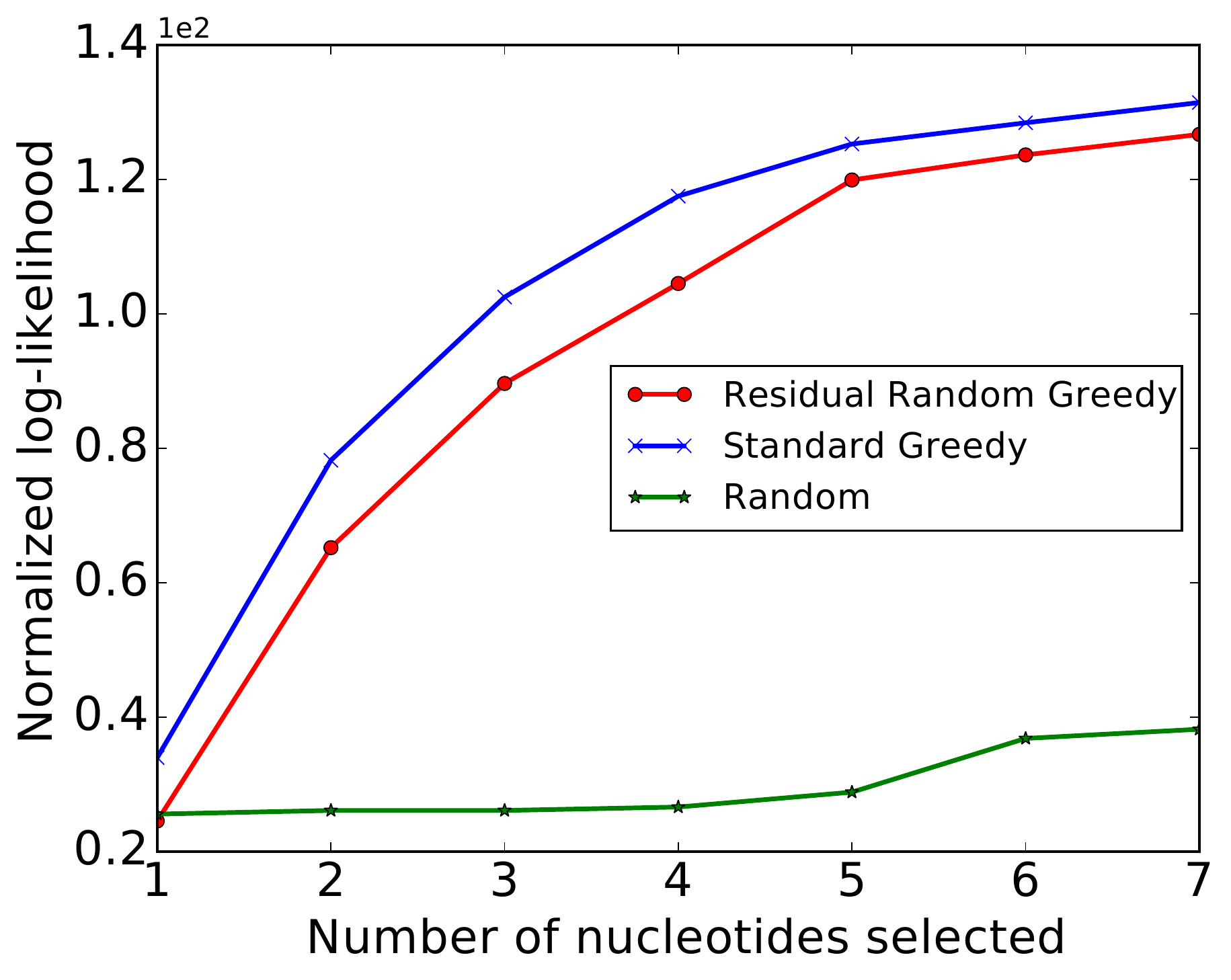}
	}
	\subfigure[Splice acceptor site detection  \label{fig:acceptor}]{
		\includegraphics[width=0.47\linewidth]{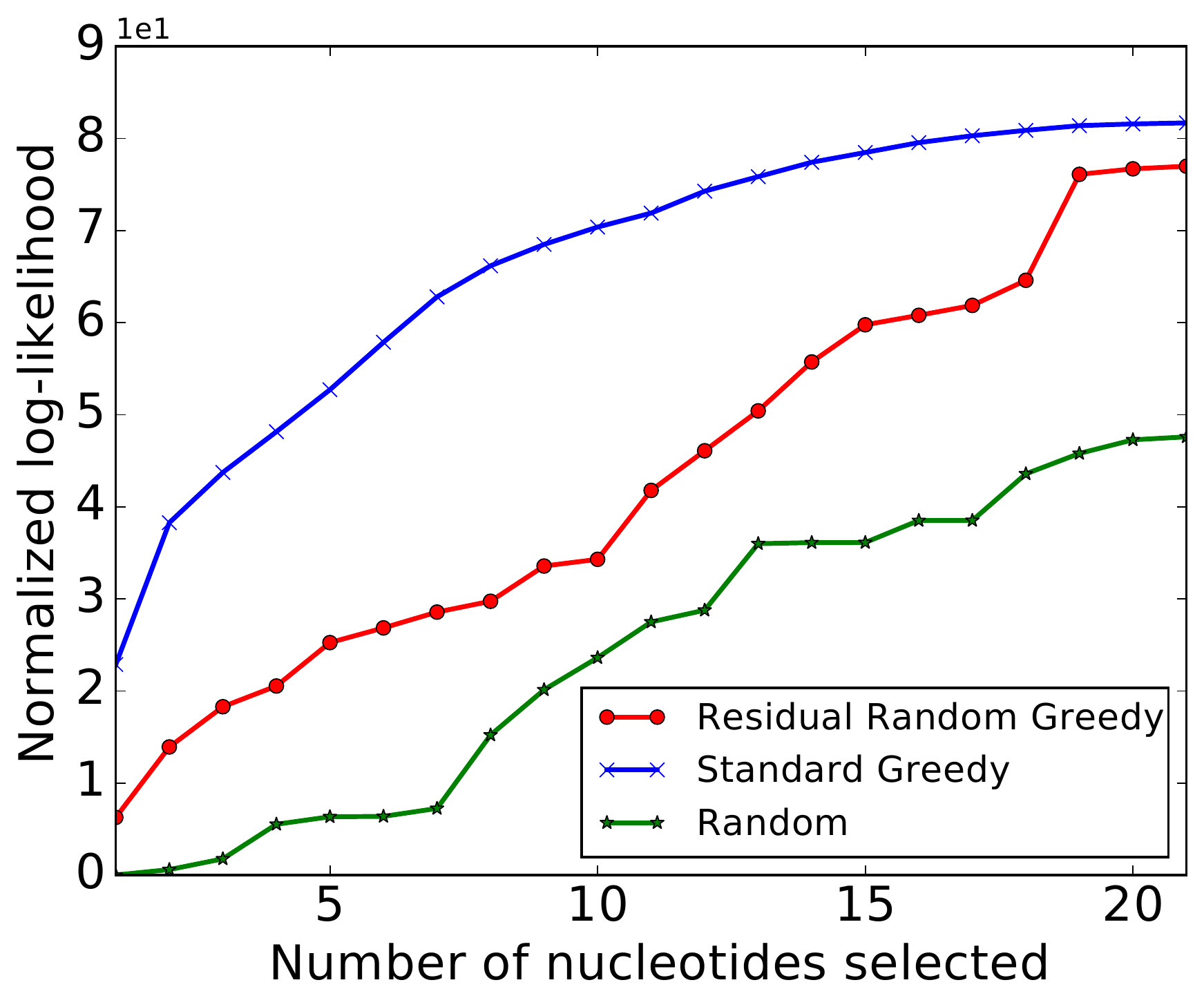}
	}
	%
	%
	\caption{Normalized log-likelihood vs.\ the number of nucleotides selected for splice donor and acceptor sites.   }\label{fig:sites}
\end{figure*}

\subsection{Black-Box Interpretation}\label{sub:interp}

In this set of experiments, we consider the problem of interpreting the predictions of black-box machine learning algorithms---\ie, explaining the reasons for their prediction. Specifically, we follow the setting of~\cite{EDFK17,LIME}. Given an image $ I $ and a label $ l $, the LIME framework~\cite{LIME} outputs the likelihood that the image $ I $ has the label $ l $. For example, the top five labels (in terms of the likelihood) assigned by the LIME framework to~\cref{fig:original_image} are \emph{Bernese mountain dog} (with likelihood 0.44), \emph{EntleBucher} (with likelihood 0.21), \emph{Greater Swiss Mountain dog} (with likelihood 0.046), \emph{Appenzeller} (with likelihood 0.033) and \emph{Egyptian cat} (with likelihood 0.0044). Here we ask which parts of the image best explain the most likely label \emph{Bernese mountain dog}; and let us denote this label by $ l_1 $ from now on. To this end, we applied the SLIC algorithm~\cite{SLIC} to the image, and this algorithm segmented the image into 25 superpixels (each superpixel is a tile of adjacent pixels of the image). Our task now is to select $ 10 $ superpixels that best explain the label $ l_1 $. We use $\cN$ to denote a ground set consisting of all the superpixels. For any subset $ S $ of $ \mathcal{N} $, let $ I(S) $ denote the subimage where only superpixels in $ S $ are present, and let $ f(S) $ be the likelihood that the subimage $ I(S) $ has the label $ l_1 $. Using this notation, our task can be formulated as the following maximization problem: $\max_{|S|\leq k} f(S)$, where $ k=10 $. We have applied \Alg, \AlgGreedy and \AlgRand to this optimization problem; and the superpixels selected by the three algorithms are visualized in~\cref{fig:interp_random_greedy,fig:interp_greedy,fig:interp_random}, respectively. We note that the set function $ f(S) $ depends on the black-box machine learning algorithm, and thus, and may not be weakly submodular, or even monotone, in general. Nevertheless, {\RRG} and our benchmark algorithms still produce interesting results when used to optimize it.

Recall that the label that we try to explain is \emph{Bernese mountain dog}. The superpixels selected by \Alg (see~\cref{fig:interp_random_greedy}) include all parts of the image that form the head of a Bernese mountain dog, while the superpixels selected by \AlgGreedy (see~\cref{fig:interp_greedy}) only cover the nose of the dog and a small portion of its body. Additionally, they also incorrectly include the head of the cat. The performance of \AlgRand is the worst (see~\cref{fig:interp_random}) as it mostly selects superpixels which are irrelevant to the dog. We also illustrate in~\cref{fig:image_interpretation} the likelihood that the subimage induced by the selected superpixels has the label $ l_1 $ versus the number of superpixels selected. It can be observed that \Alg outperforms \AlgGreedy when ten superpixels are selected. It is also noteworthy to observe that the likelihood achieved by \AlgRand remains almost zero when the number of selected superpixels varies from $ 1 $ to $ 10 $, reaching only the value $ 4.39\times 10^{-4} $ at its highest point.

\begin{figure*}[htb]
	
	\subfigure[Original image \label{fig:original_image}]{
		\includegraphics[width=0.23\linewidth]{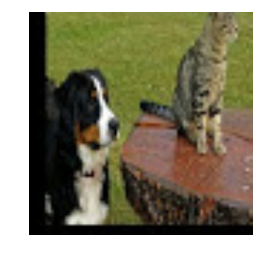}
	}	
	\subfigure[Res.~Random Greedy  \label{fig:interp_random_greedy}]{
		\includegraphics[width=0.23\linewidth]{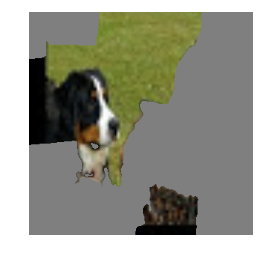}
	}
	\subfigure[Standard Greedy  \label{fig:interp_greedy}]{
		\includegraphics[width=0.23\linewidth]{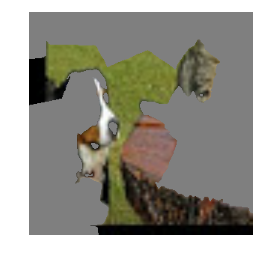}
	}\subfigure[Random \label{fig:interp_random}]{
		\includegraphics[width=0.23\linewidth]{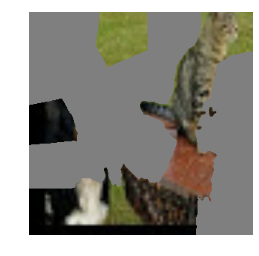}
	}
	%
	%
	\caption{Original image and visualization of   the superpixels selected by the three algorithms to explain the label \emph{Bernese mountain dog}.  }
\end{figure*}

\begin{figure}[htb]
	\centering
	\includegraphics[width=0.5\textwidth]{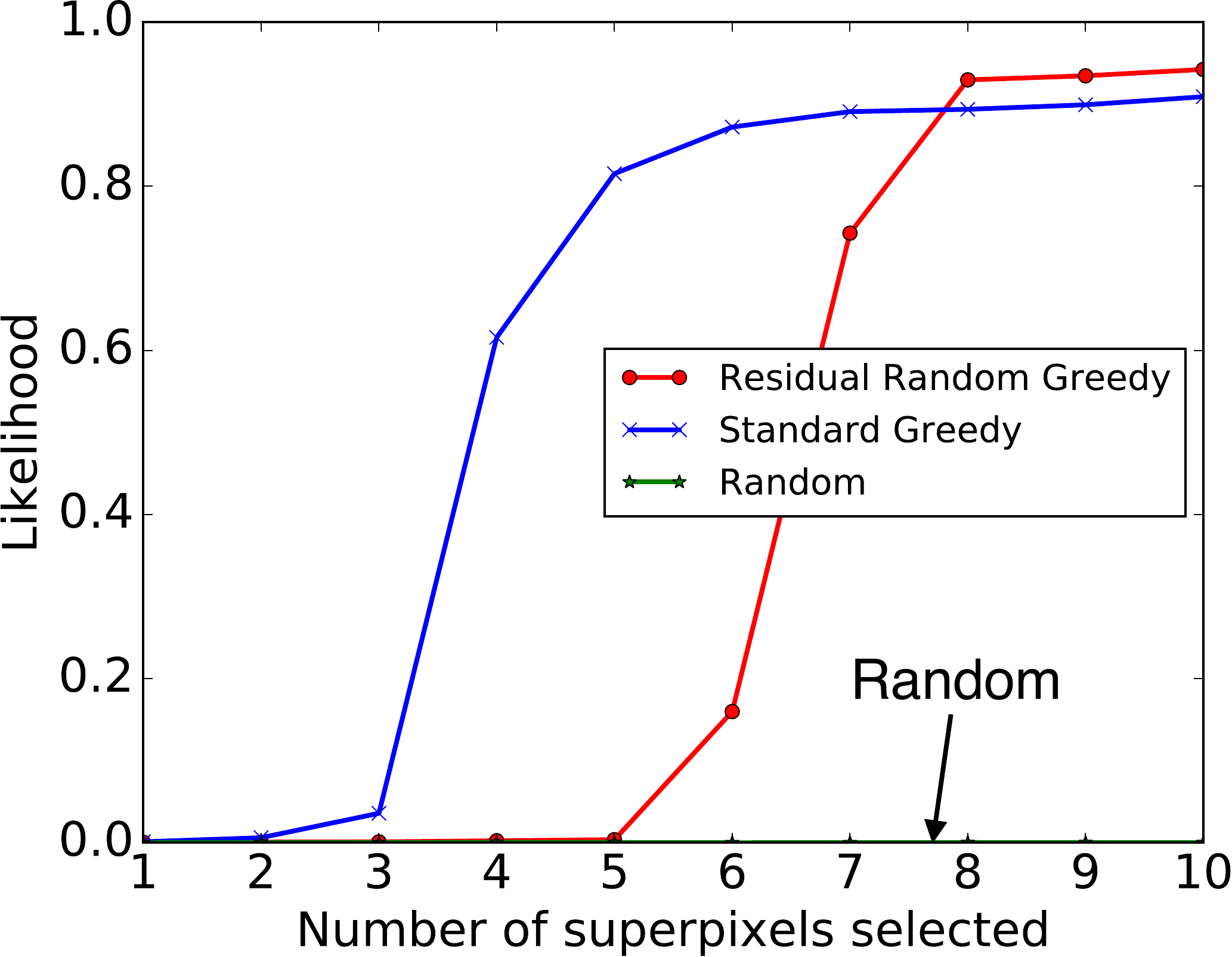}
	\caption{Likelihood that the subimage induced by the selected superpixels has the label vs.\ the number of superpixels selected. Note that the highest likelihood achieved by \AlgRand is as low as $ 4.39\times 10^{-4} $, and thus is shown to be almost zero. \label{fig:image_interpretation}}
\end{figure}

\section{Conclusion} \label{sec:conclusion}

In this paper we have proved the first non-trivial approximation ratio for maximizing a $\gamma$-weakly submodular function subject to a general matroid constraint. Our result opens the door for new applications and also suggests that the greedy algorithm performs well in practice for this problem. Moreover, we were able to demonstrate experimentally, on multiple applications, this suggested good behavior of the greedy algorithm.

The most significant question that we leave open is whether the greedy algorithm has a good \emph{provable} approximation ratio for the above problem. We note that this is not necessarily implied by the good practical behavior of the greedy algorithm. For example, on the closely related problem of maximizing a non-monotone submodular function, the greedy algorithm performs well in practice despite having an unbounded theoretical approximation ratio. Personally, we tend to believe that the greedy algorithm does have a good provable approximation ratio for the problem because we were unable to design any example on which the approximation ratio of the greedy algorithm is non-constant (for a constant $\gamma$). However, proving this formally is likely to require new ideas, and is thus, a very interesting area for future work.

\subsubsection*{Acknowledgements}

We thank Sahand Negahban for fruitful discussions.

\bibliographystyle{plain}
\bibliography{WeakSubmodular}

\appendix
\section{Removal of the Low Order Term} \label{app:low_order_term}

Theorem~\ref{thm:main_result_weak} proves an approximation ratio guarantee for Algorithm~\ref{alg:ResidualRandomGreedy} which is weaker than the approximation ratio guarantee of Theorem~\ref{thm:main_result} by a low order term of $O(k^{-1})$. In this appendix we prove that this low order term can be dropped from the guarantee of Theorem~\ref{thm:main_result_weak}, which yields Theorem~\ref{thm:main_result}. In fact, we even prove the following stronger theorem.

\begin{theorem}
Let $c(\gamma)$ be an arbitrary function of $\gamma$, and let $\eps(k)$ be a function of $k$ that approaches $0$ as $k$ increases. Then, if the approximation ratio of Algorithm~\ref{alg:ResidualRandomGreedy} is at least $c(\gamma) - \eps(k)$, then it is also at least $c(\gamma)$.
\end{theorem}
\begin{proof}
Assume towards a contradiction that the theorem is wrong. This implies that the approximation ratio of Algorithm~\ref{alg:ResidualRandomGreedy} is at least $c(\gamma) - \eps(k)$, and yet there exists an instance $I$ with a $\gamma_I$-weakly submodular objective function on which the approximation ratio of Algorithm~\ref{alg:ResidualRandomGreedy} is $c' < c(\gamma_I)$. Since $\eps(k)$ approaches $0$ when $k$ increases, we can find a value $k'$ such that $c' < c(\gamma_I) - \eps(k)$ for every $k \geq k'$. Note that we may assume, without loss of generality, that $k'$ is a non-negative integer. Consider now the variant of Algorithm~\ref{alg:ResidualRandomGreedy} given as Algorithm~\ref{alg:ResidualRandomGreedyVariant}. Intuitively, this variant extends the input by introducing $k'$ new elements which do not affect the objective function and can be used to extend every independent set of the matroid.
\begin{algorithm}
\caption{\textsf{Residual Random Greedy for Matroids (Variant)}$(f, \cM)$} \label{alg:ResidualRandomGreedyVariant}
Create $k'$ new elements, and let $\cN'$ denote the set of these new elements.\\
Extend the object function $f$ to the ground set $\cN \cup \cN'$ by setting $f(S) = f(S \setminus \cN')$ for every set $S$ which includes new elements.\\
Extend the matroid $\cM$ to the ground set $\cN \cup \cN'$ by determining that a set $S$ which includes new elements is independent if and only if $S \setminus \cN'$ is independent.

\BlankLine
Initialize: $S_0 \leftarrow \varnothing$.\\
\For{$i$ = $1$ \KwTo $k + k'$}
{
    Let $M_i$ be a base of $\cM / S_{i - 1}$ maximizing $\sum_{u \in M_i} f(u \mid S_{i - 1})$.\\
    Let $u_i$ be a uniformly random element from $M_i$.\\
    $S_i \leftarrow S_{i - 1} + u_i$.
}
Return $S_{k + k'} \setminus \cN'$.
\end{algorithm}

Observe that the extension increases the rank of the matroid $\cM$ to $k + k'$ and preserves the $\gamma$-weak submodularity of the objective function $f$. Thus, by our assumption on the approximation ratio of Algorithm~\ref{alg:ResidualRandomGreedy}, $S_{k + k'}$ must provide an approximation ratio of at least $c(\gamma) - \eps(k + k')$ for the problem of maximizing the extended objective function $f$ subject to the extended matroid $\cM$. One can verify that, together with the properties of $\cN'$, this implies that $S_{k + k'} \setminus \cN'$ provides an approximation ratio of at least $c(\gamma) - \eps(k + k')$ for the problem of maximizing the original objective function $f$ subject to the original matroid $\cM$. Hence, Algorithm~\ref{alg:ResidualRandomGreedyVariant} has an approximation ratio of at least $c(\gamma) - \eps(k + k') > c(\gamma) - [c(\gamma_I) - c']$ in general, which implies that for the specific instance $I$ the approximation ratio of Algorithm~\ref{alg:ResidualRandomGreedyVariant} is strictly better than $c'$.

The final step required for getting the contradiction that we seek is to observe that Algorithms~\ref{alg:ResidualRandomGreedy} and~\ref{alg:ResidualRandomGreedyVariant} share an identical output distribution for every given instance. Before we explain why that observation is true, let us note that it indeed implies a contradiction because Algorithm~\ref{alg:ResidualRandomGreedy} has an approximation ratio of $c'$ for the instance $I$, while Algorithm~\ref{alg:ResidualRandomGreedyVariant} has a strictly better approximation ratio for this instance. Thus, it remains to explain why Algorithms~\ref{alg:ResidualRandomGreedy} and~\ref{alg:ResidualRandomGreedyVariant} have identical output distributions, which is what we do in the rest of this paragraph. Note that Algorithm~\ref{alg:ResidualRandomGreedyVariant} must add all the elements of $\cN'$ to its solution set at some point because every base of the extended matroid $\cM$ contains all of $\cN'$. This means that we can view Algorithm~\ref{alg:ResidualRandomGreedyVariant} as a variant of Algorithm~\ref{alg:ResidualRandomGreedy} that has $k'$ more rounds, but must waste $k'$ of its rounds on adding the elements of $\cN'$ which do not affect anything and are removed at the end anyhow. Hence, the two algorithms share an identical behavior if we disregard the extra $k'$ rounds that Algorithm~\ref{alg:ResidualRandomGreedyVariant} wastes on adding elements of $\cN'$.
\end{proof}

\end{document}